\documentclass[11pt]{article}

\usepackage{times}
\usepackage[compact]{titlesec}
\usepackage{enumitem}
\setlist{nosep}
\usepackage{microtype}
\usepackage{etoolbox}
\AtBeginDocument{
  \setlength{\abovedisplayskip}{6pt}
  \setlength{\belowdisplayskip}{6pt}
  \setlength{\abovedisplayshortskip}{3pt}
  \setlength{\belowdisplayshortskip}{3pt}
}

\usepackage{graphicx}
\usepackage{amsmath}
\usepackage{amssymb}
\usepackage{braket}
\usepackage{quantikz}
\usepackage{amsthm}
\usepackage[top=0.8in, bottom=0.8in, left=0.8in, right=0.8in]{geometry}
\usepackage{xcolor}
\definecolor{lightblue}{RGB}{80,150,255}
\usepackage[
  colorlinks=true,
  linkcolor=lightblue,
  citecolor=lightblue,
  urlcolor=lightblue,
  filecolor=lightblue
]{hyperref}
\usepackage{tcolorbox}

\theoremstyle{plain}
\newtheorem{theorem}{Theorem}[section]
\newtheorem{lemma}[theorem]{Lemma}
\newtheorem{proposition}[theorem]{Proposition}
\newtheorem{corollary}[theorem]{Corollary}

\theoremstyle{definition}
\newtheorem{definition}[theorem]{Definition}
\newtheorem{example}[theorem]{Example}

\theoremstyle{remark}
\newtheorem{remark}[theorem]{Remark}
\newtheorem{note}[theorem]{Note}

\theoremstyle{definition}
\newtheorem{assumption}[theorem]{Assumption}

\newcommand{\Tr}{\operatorname{Tr}}

\title{Modeling Quantum Error Detection with Transition Matrices}

\author{
    Rohan S. Kumar\thanks{Yale Quantum Institute \& Department of Computer Science, rohan.s.kumar@yale.edu} \and
    Ben Foxman\thanks{Yale Quantum Institute \& Department of Computer Science, Yale University} \and
    Yongshan Ding\thanks{Yale Quantum Institute \& Department of Computer Science, Yale University}
}

\date{}

\begin{document}

\maketitle
\vspace{-1em}

\begin{abstract}
We construct an explicit transition matrix that exactly describes code-adapted basis-population dynamics of stabilizer codes under circuit-level stochastic Pauli noise. The matrix is expressed in a code-adapted basis that captures probability flow between logical-basis labels and syndrome sectors. For quantum error detection (QED), we incorporate post-selection by aggregating rejected outcomes into an absorbing state, so that a single transition matrix represents a full clock cycle of gates, syndrome extraction, and post-selection. Additionally, we characterize when protocol symmetries permit exact lumping to simpler, more interpretable models.

The transition matrix framework enables direct application of classical stochastic-matrix techniques to analyze multi-cycle QED protocols. First, we identify a leading-order obstruction to a stationary logical-only accepted-map description, and quantify the accepted leakage injection from QED-check imperfections that causes it. Second, we express leading-order QED efficiency as a function of check frequency in terms of physically meaningful parameters, and prove that less frequent checks improve efficiency to leading order. We evaluate the framework on the $[[4,2,2]]$ and $[[5,1,3]]$ codes under depolarizing gate noise and readout errors: transition-matrix predictions agree with million-shot Stim Monte Carlo at the scale of the reported confidence intervals, and matrix-computed results illustrate the first-cycle transient and interval-dependent efficiency. More broadly, the transition-matrix formalism provides both an analytical foundation for QED analysis and an interpretable tool for understanding code behavior under realistic noise.
\end{abstract}

\section{Introduction \label{sec:intro}}

Fault-tolerant quantum computation protects logical information via redundant encoding and repeated error detection or correction \cite{shor1995scheme,steane1996error, terhal2015quantum}. Stabilizer codes are the dominant framework for realizing this protection, using syndrome measurements to identify and address errors. A fundamental modeling question arises: at what level of granularity should we
describe the dynamics of an encoded system under noise? One option is to model the quantum channel of the full $n$-qubit
physical register, but the resulting $4^n \times 4^n$ superoperators become
unwieldy even for modest code sizes. The opposite extreme models the system purely at the logical
level, treating the encoded register as an effective $k$-qubit channel. This is analytically
convenient but implicitly assumes that physical noise manifests only as logical noise. In practice
this is a poor approximation: it misses how probability flows between the code space and syndrome
sectors, and how post-selection or correction conditions the ensemble.

In this work we develop an intermediate approach: an explicit transition matrix acting on a
code-adapted basis of the data and ancilla that separates logical states from syndrome sectors. Under
\emph{circuit-level stochastic Pauli noise}, we prove this transition matrix provides an \emph{exact}
description of basis-population dynamics. The mechanism is structural: in our code-adapted basis,
indexed by stabilizer syndrome and logical labels, both data-preserving Clifford syndrome-extraction
circuits and Pauli faults act as signed permutations and therefore commute with dephasing, implying
that populations evolve autonomously according to a classical Markov process.

Before protocol-specific reductions, the matrix is exponentially large, so this framework is not
proposed as an asymptotically efficient simulator; efficient simulation of stabilizer codes under
stochastic Pauli noise is already well established via Gottesman--Knill~\cite{aaronson2004improved}. Rather, our contribution
lies in providing an explicit, interpretable Markov kernel on operationally meaningful labels, which
makes the probability flow through a code transparent and amenable to classical stochastic-matrix
proof techniques. Related stochastic-matrix perspectives appear in prior and concurrent
work~\cite{ziyad2025emergent,kwiatkowski2025constructing,kumar2026co}.
Here we prove exact population closure on the full code-adapted label space,
formalize composable reductions for post-selected QED and interleaved gate layers,
and derive their perturbative consequences.

\begin{tcolorbox}[title=\textbf{Contributions: Framework and Tools}]
\begin{itemize}[leftmargin=*]
\item \textbf{Code-adapted transition matrix.} We construct an \textit{explicit} Markov kernel on syndrome and logical-eigenvalue labels that exactly describes population dynamics under circuit-level stochastic Pauli noise.

\item \textbf{QED reduction.} In QED, rejected runs terminate, so naive iteration on the joint state space fails to capture multi-cycle post-selection. We handle this by aggregating rejected outcomes into an absorbing state, yielding a $(2^n+1)$-dimensional matrix exact over arbitrarily many cycles.

\item \textbf{Clock-cycle composition.} When QED checks are interleaved with noisy gates, we model a full clock cycle, including gates, syndrome extraction, and post-selection, in a single matrix admitting a clean perturbative expansion.
\end{itemize}
\begin{itemize}[leftmargin=*]
\item \textbf{Symmetry-based reduction.} To obtain interpretable reduced models, we discuss how matrix symmetries can enable Markov lumping of the transition matrix. We then show that under specific implementation and noise assumptions, permutation symmetries of the code are preserved in the transition matrix, enabling exact orbit lumping. This enables reduction of the $[[4,2,2]]$ code from $17$ to $7$ states, and the Steane code from $129$ to $11$ states.
\end{itemize}
\end{tcolorbox}

Using this framework, we analyze phenomena of practical interest. Beyond their independent utility, these analyses demonstrate how our explicit construction enables new theoretical results.

\begin{tcolorbox}[title=\textbf{Contributions: Applications and Results}]
\begin{itemize}[leftmargin=*]
\item \textbf{Emergent logical non-Markovianity (Section~\ref{sec:nonmarkov}).} Conditioned logical statistics can exhibit history dependence under memoryless noise~\cite{ziyad2025emergent}. Our framework captures this effect naturally; we go further by quantifying it via an explicit accepted-leakage-injection term, proving its asymptotics, and showing intervening gate noise contributes only at second order.

\item \textbf{QED efficiency analysis (Section~\ref{sec:qed_efficiency}).} QED efficiency measures improvement in logical error rate per discarded shot~\cite{kumar2026co}. We derive a closed-form first-order expression for $\mathrm{SE}(m)$ and show that, to leading order, $\mathrm{SE}(m)$ increases with the number of gates between QED checks. The expansion yields physically interpretable parameters that admit closed-form expressions for leading-order efficiency scaling.

\item \textbf{Numerical evaluation (Section~\ref{sec:evaluation}).} We compare transition-matrix predictions with $10^6$-shot Stim Monte Carlo for the $[[4,2,2]]$ and $[[5,1,3]]$ codes, and evaluate the first-cycle transient and interval-dependent QED efficiency under depolarizing gate noise and readout errors.
\end{itemize}
\end{tcolorbox}

We hope the general framework will be extended to study other aspects of Quantum Error Detection codes and stabilizer codes in general (see Section~\ref{sec:discussion}).

\paragraph{Organization.} Section~\ref{sec:related_work} discusses related work. Section~\ref{sec:preliminaries} reviews stabilizer codes and basis-adapted channels. Sections~\ref{sec:modeling_tm}--\ref{sec:clock_cycle} construct the transition matrix, QED reduction, and clock-cycle composition. We apply the framework to logical non-Markovianity (Section~\ref{sec:nonmarkov}) and QED efficiency (Section~\ref{sec:qed_efficiency}). Section~\ref{sec:symmetry_model_reduction} develops symmetry-based reductions. Section~\ref{sec:evaluation} evaluates the construction and both applications numerically. Section~\ref{sec:discussion} concludes with extensions and practical considerations.

\section{Related Work \label{sec:related_work}}

Several concurrent and prior works address related questions; we summarize the key distinctions here, with extended discussion in Appendix~\ref{app:related_works}.

\textbf{Emergent logical non-Markovianity.}
Ziyad~\emph{et al.}~\cite{ziyad2025emergent} define logical Markovianity for error-corrected registers and show that Markovian physical operations can induce history-dependent logical statistics; they also introduce a ``stochastic picture'' in which repeated syndrome extraction under Pauli noise is modeled as a Markov chain via an explicit transition matrix (worked out for the 3-qubit repetition code).
Our work builds on this viewpoint by giving an explicit, exact, and \emph{general} kernel on a code-adapted label basis for data-preserving Clifford syndrome-extraction circuits of arbitrary stabilizer codes under circuit-level Pauli-stochastic noise, and by incorporating post-selected QED via an exact absorbing rejection state that makes multi-cycle conditioning analytically tractable.

\textbf{Approximate logical Markovian models.}
Kwiatkowski~\emph{et al.}~\cite{kwiatkowski2025constructing} prove that under certain conditions, repeated QEC cycles are well-approximated by a memoryless logical channel. Our results are complementary, quantifying a first-cycle transient in post-selected QED caused by accepted leakage. Our work is procedurally distinct as we retain an explicit finite-state description to analyze finite-time transients, accepted-shot lifetimes, and post-selection effects in QED.

\textbf{Resource destroying maps.}
Liu, Hu, and Lloyd~\cite{liu2017resource} introduced resource destroying maps (RDMs) to characterize resource-free operations. We adopt this framework with dephasing in a code-adapted basis, providing explicit constructions tailored to fault-tolerant stabilizer protocols. In particular, we use monomial Kraus representations as a simple certificate of $\mathcal{B}$-commutation for the noisy stabilizer measurement circuits we consider, which makes the induced transition matrices explicit and composable.

\textbf{QED--PEC co-design.}
Kumar~\emph{et al.}~\cite{kumar2026co} study QED--PEC co-design
using transition-matrix models and discuss the physical origin of first-cycle transients.
Our work provides theoretical foundations for these models by proving exact
population-level closure under stochastic Pauli noise, and derives the transient's
leading-order magnitude, showing that only QED-check noise contributes at first order.

\textbf{Optimizing check frequency.}
Abu-Nada, Fortescue, and Byrd~\cite{abunada2015optfreq} study how often to apply fault-tolerant QEC operations, showing that correcting after every gate need not be optimal when the correction step itself introduces errors.
Our interval results are complementary but distinct, analyzing post-selected QED and the QED-efficiency metric directly through our transition matrix formalism.

\section{Preliminaries \label{sec:preliminaries}}

Quantum error correction (QEC) enables reliable quantum computation with imperfect, \textit{physical qubits} by encoding them into higher-dimensional \textit{logical qubits}. Logical error rates are suppressed using a process called \textit{syndrome extraction}, where errors on physical qubits are repeatedly measured and corrected. When the physical error rate is below a code-dependent threshold, logical error rates can be arbitrarily suppressed ~\cite{shor1995scheme,steane1996error,aliferis2006accuracy,terhal2015quantum}.

Quantum error detection (QED) is a closely related paradigm in which syndromes are used to detect, rather than immediately correct, errors. For suitable circuits, QED codes can suppress errors in accepted runs with low qubit overhead, at the cost of increased sampling overhead~\cite{self2024protecting,jin2025iceberg}. These tradeoffs make QED an attractive option for near-term devices that lack the resources for full fault-tolerant error correction.

Most QED and QEC codes are \textit{stabilizer codes}, which implement syndrome extraction using stabilizer measurements and Clifford-dominated control~\cite{kitaev2003fault,fowler2012surface}. An $[[n,k,d]]$ stabilizer code encodes $k$ logical qubits into $n$ physical qubits with code distance $d$, defined by an abelian subgroup $\mathcal{S} \subset \mathcal{P}_n$ with $n-k$ independent generators; the code space is the simultaneous $+1$ eigenspace of all elements of $\mathcal{S}$, and any Pauli error of weight less than $d$ either acts trivially on the code space or has a nontrivial syndrome and can therefore be detected.

\subsection{The Stabilizer Formalism}
Let $\mathcal{P}_n$ denote the signed Pauli group on $n$ qubits. Throughout, stabilizer generators and logical Pauli representatives are chosen Hermitian (phase $\pm1$), so their eigenvalues are $\pm1$. In this work, we consider $[[n,k,d]]$ \textit{stabilizer codes}, which are specified by an abelian subgroup $\mathcal{S}=\langle S_1,\dots,S_{n-k}\rangle\subset\mathcal{P}_n$ which does not contain $-I$. The \textit{code space} is the joint $+1$ eigenspace of all operators in $\mathcal S$, i.e.\ the space spanned by \textit{codewords} $\ket{\psi}$ satisfying 
\begin{equation}
    P\ket{\psi} = \ket{\psi}, \;\; \forall P \in \mathcal{S}.
\end{equation}
For an $[[n,k,d]]$ code, the joint $+1$ eigenspace defines a $2^k$-dimensional code space~\cite{gottesman1997stabilizer,calderbankshor1996good}. To any code, we can associate a set of \textit{logical Pauli operators}, given by representatives of the quotient $(N(\mathcal S)\cap\mathcal P_n)/\mathcal S$. 
Measuring all $n-k$ stabilizer generators yields a syndrome $s\in\{0,1\}^{n-k}$, with the post-measurement state satisfying 
\begin{equation}
    (\forall i) \ \ \ \ S_i\ket{\phi} = (-1)^{s_i}\ket{\phi}
\end{equation}
Finally, throughout this work, \emph{leakage} refers to channels through which codestates transition to states outside the logical code space but still within the $n$-qubit
Hilbert space $(\mathbb{C}^2)^{\otimes n}$. This is distinct from the notion of \emph{physical
leakage} where individual qubits leave the two-level computational subspace entirely.

\subsection{Basis-Adapted States and Channels}
\label{sec:prelim_stoch_sim}
We now discuss when the action of quantum channels are describable as transition matrices over population vectors over a basis of Hilbert space. To do so, we adapt several notions defined in~\cite{liu2017resource}. 

Fix a Hilbert space $\mathcal{H}\cong\mathbb{C}^N$ with orthonormal basis $B=\{\ket{b_i}\}_{i=1}^N$. Let $\Pi_i := \ket{b_i}\!\bra{b_i}$. The \emph{$B$-dephasing channel} is given by
\begin{equation}
\Delta_B(\rho) := \sum_{i=1}^N \Pi_i \rho \Pi_i.
\end{equation}
States which are diagonal in the basis $B$ satisfy  $\Delta_B(\rho)=\rho$. We call such states \emph{$B$-incoherent}.

\begin{definition}[Nonactivating and Commuting Channels~\cite{liu2017resource}]
\label{def:RDM_classes}
A completely positive, trace non-increasing map $\mathcal{E}$ is
\begin{align}
\textit{$B$-nonactivating}
\quad\Longleftrightarrow\quad
\Delta_B\circ\mathcal{E}
&=
\Delta_B\circ\mathcal{E}\circ\Delta_B,
\label{eq:Bnonact_def}\\
\textit{$B$-commuting}
\quad\Longleftrightarrow\quad
\Delta_B\circ\mathcal{E}
&=
\mathcal{E}\circ\Delta_B.
\label{eq:Bcomm_def}
\end{align}
\end{definition}
Every $B$-commuting map is $B$-nonactivating.  The complementary notion $B$-nongenerating and its
relationship to these conditions are summarized in Appendix~\ref{app:rdm_terminology}.
\begin{definition}[$B$-population vector]
Define the $B$-population vector $p_B(\rho)\in\mathbb{R}^N$ by
\begin{equation}
(p_B(\rho))_i := \Tr(\Pi_i\rho).
\label{eq:B_pop_vec}
\end{equation}
\end{definition}

\begin{definition}[Induced $B$-transition matrix]
Given a completely positive, trace non-increasing map $\mathcal{E}$, define the transition matrix $T_B(\mathcal{E})\in\mathbb{R}^{N\times N}$ by
\begin{equation}
\bigl(T_B(\mathcal{E})\bigr)_{i,j} := \Tr\!\bigl(\Pi_i\,\mathcal{E}(\Pi_j)\bigr).
\label{eq:TB_def}
\end{equation}
The transition matrix is column-stochastic if $\mathcal{E}$ is trace-preserving, and column-substochastic otherwise.
\end{definition}

\begin{proposition}[$B$-nonactivation $\Rightarrow$ exact population evolution via a transition matrix]
\label{prop:nonact_implies_TM}
If $\mathcal{E}$ is $B$-nonactivating, then for every state $\rho$ and all $t\in\mathbb{N}$,
\begin{equation}
p_B\!\bigl(\mathcal{E}^{\,t}(\rho)\bigr) = \bigl(T_B(\mathcal{E})\bigr)^{t}\,p_B(\rho).
\label{eq:TB_exact_iter}
\end{equation}
\end{proposition}

\begin{proof}[Proof sketch]
$B$-nonactivation gives $\Delta_B\!\circ\!\mathcal{E}=\Delta_B\!\circ\!\mathcal{E}\!\circ\!\Delta_B$,
so for each $i$,
\begin{equation}
    (p_B(\mathcal{E}(\rho)))_i=\Tr(\Pi_i\mathcal{E}(\rho))
=\Tr(\Pi_i\mathcal{E}(\Delta_B(\rho)))
=\sum_j \Tr(\Pi_j\rho)\,\Tr(\Pi_i\mathcal{E}(\Pi_j)).
\end{equation}

Thus $p_B(\mathcal{E}(\rho))=T_B(\mathcal{E})\,p_B(\rho)$, and iterating yields
$p_B(\mathcal{E}^t(\rho))=T_B(\mathcal{E})^t p_B(\rho)$.
\end{proof}

The full proof appears in Appendix~\ref{app:prelim_proofs}.

\begin{remark}
    A $B$-nonactivating channel is equivalent to the statement that the population vector of the output state depends only on a (sub)stochastic matrix applied to the population vector of the input state (see Proposition~\ref{prop:nonact_implies_TM}). In this work we will verify that stabilizer codes under Pauli Stochastic noise satisfy the stronger $B$-commuting condition. When a quantum channel is $B$-commuting, any B-incoherent initial state remains B-incoherent under iteration, so on that sector the full density operator is determined by its population vector; see Appendix~\ref{app:rdm_terminology} for more information. Accordingly, throughout we focus on cycle-boundary ensembles that are $\mathcal{B}$-incoherent (i.e., diagonal in $\mathcal{B}$), so the transition-matrix model yields exact values for all diagonal-in-$\mathcal{B}$ observables we study (e.g., acceptance and ASC).
\end{remark}

\begin{proposition}[Closure of $B$-nonactivating and $B$-commuting classes]
    The $B$-nonactivating and $B$-commuting classes are closed under convex combination~\cite{liu2017resource}.
\label{prop:bcomm_closure}
\end{proposition}

When proving that channels are $B$-commuting, a useful definition is that of a \emph{Monomial operator}.

\begin{definition}[$B$-Monomial operator]
An operator $K$ is \emph{monomial in basis $B=\{\ket{b_i}\}$} if
$K\ket{b_i} = c_i\ket{b_{\pi(i)}}$ for some permutation $\pi$ and nonzero scalars $c_i$.
Equivalently, the matrix of $K$ in $B$ has exactly one nonzero entry in each row and column.
\end{definition}

\begin{lemma}[$B$-Monomial channels commute with $B$-dephasing]
\label{lem:monomial_commutes_dephase}
If a completely positive, trace non-increasing map $\mathcal{E}$ admits a Kraus representation in which every Kraus operator is monomial in B, then $\mathcal{E}$ is B-commuting: $\Delta_B\circ\mathcal{E} = \mathcal{E}\circ\Delta_B$. Consequently, $\mathcal{E}$ is
$B$-nonactivating (Definition~\ref{def:RDM_classes}).
\end{lemma}

\begin{proof}[Proof sketch]
Let $\mathcal{E}(\rho)=\sum_\alpha K_\alpha\rho K_\alpha^\dagger$ with each $K_\alpha$ monomial in
$B=\{\ket{b_i}\}$.  It suffices to check $\Delta_B\circ\mathcal{E}=\mathcal{E}\circ\Delta_B$ on matrix
units $\ket{b_i}\!\bra{b_j}$.  Monomiality implies
$K_\alpha\ket{b_i}\!\bra{b_j}K_\alpha^\dagger \propto
\ket{b_{\pi_\alpha(i)}}\!\bra{b_{\pi_\alpha(j)}}$.  If $i=j$ the output is diagonal and unchanged by
$\Delta_B$; if $i\neq j$ it is off-diagonal and annihilated by $\Delta_B$.  Hence
$\Delta_B\circ\mathcal{E}=\mathcal{E}\circ\Delta_B$.
\end{proof}

The full proof appears in Appendix~\ref{app:monomial_commutes}.

\begin{remark}[$B$-commutation vs Gottesman--Knill simulatability]
The criterion ``$\mathcal{E}$ commutes with a fixed dephasing map $\Delta_B$'' is basis-dependent and
is neither necessary nor sufficient for efficient stabilizer simulation.  For example, the Hadamard
gate is Clifford, but is not monomial in the computational basis, and so
does not commute with computational dephasing. Conversely, classical reversible gates such as
Toffoli are monomial in the computational basis but are not Clifford.
In this work we use monomiality as a certificate of population closure in a
specific code-adapted basis~\eqref{eq:joint_basis_def}.
\end{remark}

\section{Modeling Stabilizer Codes as Transition Matrices \label{sec:modeling_tm}}

\subsection{A Code-Specific Basis}
\label{sec:code_specific_basis}
We now construct the code-adapted basis that will serve as the state space for our transition matrix model. Given a stabilizer code $\mathcal{C}$ and a choice of $k$ independent, mutually commuting logical Pauli operators $\{\overline{P}_1,\dots,\overline{P}_k\}\subset N(\mathcal{S})$ that are independent modulo $\mathcal{S}$, we define the basis as follows. First, define a basis $\{\ket{\ell}_L\}_{\ell\in\{0,1\}^k}$ for the code space satisfying
\begin{equation}
\overline{P}_j\ket{\ell}_L = (-1)^{\ell_j}\ket{\ell}_L.
\end{equation}
Then, for each syndrome $s \in \{0, 1\}^{n-k}$, choose a Pauli representative $E_s$ such that applying $E_s$ to a codestate produces that syndrome, i.e., maps the codestate to the $s$-th coset of the code space. The full code-adapted basis is
\begin{equation}
\ket{s,\ell} := E_s\ket{\ell}_L, \qquad s\in\{0,1\}^{n-k},\ \ell\in\{0,1\}^k.
\label{eq:code_adapted_basis_def}
\end{equation}
\noindent For $s\neq 0^{n-k}$, the label $\ell$ indicates which codespace vector $\ket{\ell}_L$ was acted upon by $E_s$, rather than the eigenvalues of $\{\overline{P}_j\}$ on $\ket{s,\ell}$ (which may differ if $E_s$ anticommutes with some logical operators). This is purely a labeling convention and does not affect the transition-matrix structure. These $2^n$ vectors form an orthonormal basis $\mathcal{B}_{\mathcal{C}}$ of the $n$-qubit physical Hilbert space.
When the code uses $a$ ancilla qubits for syndrome extraction, with computational basis $\{\ket{y}: y\in\{0,1\}^a\}$, the joint code-adapted basis is
\begin{equation}
\mathcal{B} = \{\ket{s,\ell}\otimes\ket{y}\}, \qquad s\in\{0,1\}^{n-k},\ \ell\in\{0,1\}^k,\ y \in \{0, 1\}^a.
\label{eq:joint_basis_def}
\end{equation}

\subsection{Syndrome Extraction Circuits under Stochastic Pauli Noise are $\mathcal{B}$-commuting}
\begin{lemma}
\label{lem:paulis_monomial}
Every Pauli $P\in\mathcal{P}_{n+a}$ is monomial in the joint basis $\mathcal{B}$.
\end{lemma}
\begin{proof}
Up to phase, write $P=P_D\otimes P_A$ with $P_D$ acting on data qubits and $P_A$ on ancillas. Fix a data basis vector $\ket{s,\ell}=E_s\ket{\ell}_L$. The product $P_D E_s$ has some definite syndrome $s'$, and $E_{s'}^\dagger P_D E_s$ has trivial syndrome and therefore lies in the normalizer $N(\mathcal{S})$ up to phase. Hence $P_D E_s = \omega\,E_{s'}\,L\,S$ for some phase $\omega\in\{\pm 1,\pm i\}$, some logical Pauli representative $L\in N(\mathcal{S})\cap\mathcal{P}_n$, and some stabilizer $S\in\mathcal{S}$. The stabilizer fixes the codestate and $L$ permutes $\{\ket{\ell}_L\}$ up to phase, so $P_D\ket{s,\ell}=\omega'\ket{s',\ell'}$ for some phase $\omega'$ and label $(s',\ell')$. On the ancilla register, every Pauli maps computational basis states to computational basis states up to phase. Tensoring, every $P\in\mathcal{P}_{n+a}$ acts as a signed permutation on $\mathcal{B}$, hence is monomial in $\mathcal{B}$. The full proof appears in Appendix~\ref{app:paulis_monomial}.
\end{proof}

\begin{definition}[Data-preserving syndrome-extraction unitary]
\label{def:data_preserving}
An ideal syndrome-extraction unitary $U_0$ on $n+a$ qubits is \emph{data-preserving} if there exists a syndrome readout function $g:\{0,1\}^{n-k}\to\{0,1\}^a$ such that, for every joint basis state $\ket{s,\ell}\otimes\ket{y}\in\mathcal{B}$,
\[
    U_0\,\ket{s,\ell}\otimes\ket{y} \;=\; e^{i\phi(s,\ell,y)}\,\ket{s,\ell}\otimes\ket{y\oplus g(s)}
\]
for some phase $\phi(s,\ell,y)\in\mathbb{R}$.
\end{definition}

\begin{lemma}
\label{lem:ideal_syn_monomial}
The ideal data-preserving syndrome extraction circuit is monomial in the joint basis $\mathcal{B}$.
\end{lemma}
\begin{proof}
By Definition~\ref{def:data_preserving}, $U_0$ leaves each data label $(s,\ell)$ unchanged up to phase and updates the ancilla deterministically as $y\mapsto y\oplus g(s)$. Both maps are bijections on their respective labels, so their tensor product is a signed permutation on $\mathcal{B}$, hence $U_0$ is monomial in $\mathcal{B}$.
\end{proof}

\begin{remark}
Individual gates in the syndrome-extraction circuit need not be $\mathcal{B}$-monomial; for example, the Hadamard on ancillas is not monomial in the computational basis. This does not contradict the \emph{composite} action of the full circuit being $\mathcal{B}$-monomial.
\end{remark}

\begin{corollary}[Fault paths act as signed permutations on $\mathcal{B}$]
\label{cor:fault_path_signed_perm}
Let $\sigma$ be a stochastic Pauli fault pattern, and let $P(\sigma)\in\mathcal{P}_{n+a}$ denote the effective output-side Pauli that $\sigma$ propagates to through the Clifford circuit. The faulty unitary $U_\sigma:=P(\sigma)\,U_0$ acts as a signed permutation on $\mathcal{B}$: there exist a permutation $\pi_\sigma$ of $\mathcal{B}$-indices and phases $\theta_{\sigma,j}\in\mathbb{R}$ such that
\[
    U_\sigma\ket{b_j} \;=\; e^{i\theta_{\sigma,j}}\ket{b_{\pi_\sigma(j)}} \qquad \text{for every }\ket{b_j}\in\mathcal{B}.
\]
\end{corollary}
\begin{proof}
By Lemmas~\ref{lem:paulis_monomial} and~\ref{lem:ideal_syn_monomial}, both $P(\sigma)$ and $U_0$ are signed permutations on $\mathcal{B}$. Signed permutations are closed under composition.
\end{proof}

\begin{lemma}
    A data-preserving Clifford syndrome-extraction circuit for any stabilizer code is $\mathcal{B}$-commuting under stochastic Pauli noise.
\label{lem:syn_b_comm}
\end{lemma}
\begin{proof}
Under stochastic Pauli noise, the noisy syndrome extraction channel can be written as
\[\mathcal{E}(\rho)=\sum_{\sigma} p_\sigma\,U_\sigma\,\rho\,U_\sigma^\dagger,\qquad U_\sigma := P(\sigma)\,U_0,\]
where each $U_\sigma$ is a fault-conditioned unitary. By Corollary~\ref{cor:fault_path_signed_perm}, each $U_\sigma$ is monomial in $\mathcal{B}$, so each Kraus operator $K_\sigma:=\sqrt{p_\sigma}\,U_\sigma$ is monomial in $\mathcal{B}$, and hence $\mathcal{B}$-commuting by Lemma~\ref{lem:monomial_commutes_dephase}. The result follows from Proposition~\ref{prop:bcomm_closure}, since convex combinations preserve $\mathcal{B}$-commutation.
\end{proof}

\subsection{Constructing the Noisy Syndrome Extraction Transition Matrix}
\label{sec:construct_noisy_TM}
By Lemma~\ref{lem:syn_b_comm} and Proposition~\ref{prop:nonact_implies_TM}, a syndrome extraction circuit under stochastic Pauli noise has exact population dynamics given by a $\mathcal{B}$-transition matrix. We now give an explicit construction for this transition matrix.
Related fault-path averaging constructions appear in~\cite{kumar2026co};
here the construction is obtained as the induced $\mathcal{B}$-transition
matrix of the noisy channel, with exact population closure following
from Proposition~\ref{prop:nonact_implies_TM}.

Let $\mathcal{U}_0(\rho) := U_0 \rho U_0^\dagger$ denote the ideal syndrome extraction channel. The transition matrix for the ideal circuit is
\begin{equation}
T_0' := T_\mathcal{B}(\mathcal{U}_0), \qquad (T_0')_{i,j} = \Tr(\Pi_i\,\mathcal{U}_0(\Pi_j)),
\label{eq:T0prime_def}
\end{equation}
where $\Pi_j=\ket{b_j}\!\bra{b_j}$ and we use column-stochastic conventions.

Next, we model the impact of stochastic Pauli noise. Throughout this work, we consider circuit-level stochastic Pauli noise: Pauli errors occurring independently at discrete locations in the circuit (for example, around gates). Let $\mathcal{L}$ be the set of $F=|\mathcal{L}|$ fault locations, and let $\Lambda(\varepsilon)$ specify a Pauli-stochastic noise model of strength $\varepsilon$. For each fault location $\ell$ acting on $r_\ell$ qubits, let $\mathcal{P}_{r_\ell}^{\mathrm{proj}}$ denote the projective Pauli set on those qubits (the quotient of $\mathcal{P}_{r_\ell}$ by the global phases $\{\pm 1,\pm i\}$). A fault pattern $\sigma=(P_\ell)_{\ell\in\mathcal{L}}$, with $P_\ell\in\mathcal{P}_{r_\ell}^{\mathrm{proj}}$, occurs with probability $\Pr_{\Lambda(\varepsilon)}(\sigma)$. Since Clifford circuits conjugate Paulis to Paulis, any fault pattern $\sigma$ propagates to an effective Pauli $P(\sigma)\in\mathcal{P}_{n+a}$ at the circuit output. Let $M_{P(\sigma)}$ denote the permutation matrix on $\mathcal{B}$-indices induced by $P(\sigma)$, i.e., $(M_{P(\sigma)})_{i,j} = 1$ if $P(\sigma)\ket{b_j} \propto \ket{b_i}$ and $0$ otherwise. The fault-conditioned transition matrix is $T'(\sigma) := M_{P(\sigma)}\,T_0'$, and averaging yields the noisy transition matrix:
\begin{equation}
T'_{\Lambda(\varepsilon)} := \sum_{\sigma} \Pr_{\Lambda(\varepsilon)}(\sigma)\,M_{P(\sigma)}\,T_0'.
\label{eq:Tprime_Lambda_def}
\end{equation}
By Proposition~\ref{prop:nonact_implies_TM}, this provides an \emph{exact} population-level model.

For the perturbative results below, we assume that each nonidentity Pauli
$P$ at a fault location $\ell$ has probability $\varepsilon c_{\ell,P}$
for check noise, or $\eta c_{\ell,P}$ for gate-layer noise, with fixed
nonnegative coefficients $c_{\ell,P}$. The identity probability is the
remainder needed for normalization.

When the number of fault locations is large, restricting the sum~\eqref{eq:Tprime_Lambda_def} to patterns of weight $|\sigma|\le z$ and renormalizing by $Z_z := \sum_{|\sigma|\le z}\Pr_{\Lambda(\varepsilon)}(\sigma)$ yields a truncation correct through $z$th order with $O(F^z)$ enumeration. For i.i.d.\ fault rate $\varepsilon$, the discarded tail satisfies $1-Z_z=O(\varepsilon^{z+1})$, keeping the transition matrix interpretable and buildable even when an exhaustive sum over fault paths is infeasible. Additional practical-construction options (e.g.\ empirically estimated entries) appear in Section~\ref{sec:discussion_practical_tm}.

\begin{remark}[What the transition matrix tracks]
\label{rem:what_TM_tracks}
The construction above does not require the input state to be diagonal in $\mathcal{B}$. Coherences in $\mathcal{B}$ may be present, but under the noisy syndrome-extraction channels considered here they cannot affect later $\mathcal{B}$-basis populations. Thus $T'_{\Lambda(\varepsilon)}$ exactly predicts acceptance probabilities, rejection probabilities, syndrome/logical-label distributions, and all observables diagonal in $\mathcal{B}$ for arbitrary input states. The matrix does not, however, track coherences themselves, so it should not be used to compute observables with off-diagonal matrix elements in $\mathcal{B}$, such as fidelity to a general coherent logical superposition.
\end{remark}

\begin{note}
    Because global phases do not affect conjugation, $M_{P(\sigma)}$ depends only on
$P(\sigma)$ modulo $\{\pm1,\pm i\}$ (i.e., on the projective Pauli).
Equivalently, while $P(\sigma)$ is a monomial (signed-permutation) operator in $\mathcal{B}$,
the induced transition matrix on populations is the corresponding $0$--$1$ permutation matrix.

\end{note}

\begin{note}
    Henceforth, we drop the $\Lambda$ for convenience when the specific form of the Stochastic Pauli noise is not relevant, and write $T'_\varepsilon$ for $T'_{\Lambda(\varepsilon)}$.
\end{note}

\subsection{Constructing the Noisy QED Transition Matrix \label{sec:qed_lumping}}

The transition matrices from Section~\ref{sec:construct_noisy_TM} describe the action of the syndrome extraction circuit for any code. However, iterating this matrix does not model QED, since it fails to specify what happens to ancillas between cycles. In QED, rejected branches are \emph{discarded} between successive syndrome extraction circuits. Without taking this into account, rejected branches at the end of one QED cycle can leak back into the code space in subsequent cycles. In order to accurately model post-selection in our formalism, we make rejection terminal and aggregate rejected outcomes to obtain the exact reduced QED matrix $T_\varepsilon$.
An absorbing rejection state was also used to model post-selected QED
in~\cite{kumar2026co}; here we formalize this reduction on the full
code-adapted state space and prove its exact multi-cycle population dynamics.

\begin{definition}[Reduced QED state space]
The \emph{reduced} QED state space adjoins a single absorbing rejection state to the code-adapted basis:
\begin{equation}
\overline{\mathcal{Q}} \;:=\; \mathcal{B}_{\mathcal{C}} \cup \{R\},
\qquad
|\overline{\mathcal{Q}}| = 2^n+1.
\end{equation}
\end{definition}

\begin{definition}[Reduction map on transition matrices]
\label{def:reduction_map}
For any column-stochastic matrix $T'$ on $\mathcal{B}$, the \emph{QED reduction} $\mathcal{R}(T')$ is the $(2^n+1)\times(2^n+1)$ matrix $T:=\mathcal{R}(T')$ on $\overline{\mathcal{Q}}$ defined by:
\begin{align}
T[q',q] &:= T'[(q',0^a),(q,0^a)] \qquad (q,q'\in\mathcal{B}_{\mathcal{C}}), \label{eq:reduce_accept_block} \\
T[R,q] &:= 1-\sum_{q'} T'[(q',0^a),(q,0^a)] \qquad (q\in\mathcal{B}_{\mathcal{C}}), \label{eq:reduce_reject_row} \\
T[R,R] &:= 1, \quad T[q,R] := 0 \quad (q\in\mathcal{B}_{\mathcal{C}}). \label{eq:reduce_absorb}
\end{align}
Where the joint basis is indexed with pairs $(q,y)$ with $q\in\mathcal{B}_{\mathcal{C}}$ and $y\in\{0,1\}^a$.
\end{definition}
Intuitively, this lumping collapses all rejected outcomes into a single terminal state $R$, so the reduced matrix tracks only probability flow among accepted configurations. Equation~\eqref{eq:reduce_accept_block} retains probability mass that both lands on data label $q'$ and yields $0^a$ (accepted). Equation~\eqref{eq:reduce_reject_row} aggregates complementary mass into $R$. Equation~\eqref{eq:reduce_absorb} makes rejection absorbing.

\begin{proposition}[Exact reduced Markov model for post-selected QED]
\label{prop:qed_reduction_exact}
Let $T'$ be the unreduced transition matrix for one noisy syndrome extraction cycle. Under the operational QED procedure that reinitializes ancillas to $0^a$, applies the noisy circuit, measures ancillas, and discards unless the outcome is $0^a$, the population update on $\overline{\mathcal{Q}}$ is exactly $v \mapsto \mathcal{R}(T')v$. Iteration gives exact accepted-shot statistics: $v^{(t)} = (\mathcal{R}(T'))^t v^{(0)}$.
\end{proposition}

\begin{proof}[Proof sketch]
Consider a basis state at a cycle boundary. If accepted with data label $q$, ancillas reinitialize to $0^a$, so the joint starting label is $(q,0^a)$. After applying $T'$, the probability to reach $(q',y)$ is $T'[(q',y),(q,0^a)]$. Post-selection accepts iff $y=0^a$, passing data label $q'$ to the next cycle; otherwise the trajectory terminates in $R$. Equations~\eqref{eq:reduce_accept_block}--\eqref{eq:reduce_absorb} encode precisely these transitions, including that $R$ is absorbing. Linearity extends the result to arbitrary population vectors, and iteration follows by induction. See Appendix~\ref{app:qed_reduction_exact} for details.
\end{proof}

Using this reduction, the noisy and noiseless QED-cycle matrices are:
\begin{equation}
T_{\Lambda(\varepsilon)} := \mathcal{R}(T'_{\Lambda(\varepsilon)}), \qquad T_0 := \mathcal{R}(T_0').
\label{eq:T_reduced_defs}
\end{equation}

\begin{lemma}[Noiseless QED Matrix]
    Assume the QED readout is \emph{complete}: $g(s)=0^a$ iff $s=0^{n-k}$, where $g$ is the syndrome readout function of Definition~\ref{def:data_preserving}. For an ideal QED cycle, ordering states as (logical $L$) $\oplus$ (leakage $X$) $\oplus$ (rejection $R$) with $L=2^k$ and $X=2^n-2^k$, the noiseless matrix has block form:
\begin{equation}
T_0 = \begin{bmatrix} I_L & 0 & 0 \\ 0 & 0 & 0 \\ 0 & \mathbf{1}_X^T & 1 \end{bmatrix}.
\label{eq:T0_block_form}
\end{equation}
\end{lemma}
\begin{proof}
By readout completeness, zero-syndrome (logical) states yield ancilla outcome $0^a$ and are accepted, while every nonzero-syndrome (leakage) state yields a nonzero ancilla outcome and is rejected. The rejection state is terminal by Definition~\ref{def:reduction_map}.
\end{proof}

\section{Modeling a Full Clock Cycle \label{sec:clock_cycle}}

In practice, QED checks are interleaved with logical gates. When modeling the added effect of noisy gates, one desires a matrix that describes the \emph{full clock cycle}: a noisy layer followed by a noisy QED check. We now discuss the regimes in which modeling this with a transition matrix is accurate.

\subsection{Noise Assumptions \label{sec:noise_ass}}

Since ideal gates need not be $\mathcal{B}$-nonactivating, most ideal gates cannot be modeled as transition matrices. Therefore, in our analysis, we model only the noise component of gate layers, not the ideal logical operations themselves. This suffices for characterizing how a code behaves under noise without requiring a full simulation of logical evolution.

For this approach to be valid, the aggregate noise, when commuted through both ideal gate layers and the ideal syndrome extraction circuit, must be $\mathcal{B}$-nonactivating. Stochastic Pauli noise satisfies the stronger $\mathcal{B}$-commuting condition (Lemma~\ref{lem:syn_b_comm}), so it suffices to verify that the aggregate noise remains stochastic Pauli. This holds when the ideal gates are Clifford, since Clifford conjugation maps Paulis to Paulis. We therefore adopt the assumption of \emph{Clifford gates under stochastic Pauli noise} throughout this work. Outside this regime, the same state space can be used to build an empirical Markov model from cycle-boundary transition statistics with rejection lumped as in Section~\ref{sec:qed_lumping}, but the exactness guarantees of Lemma~\ref{lem:syn_b_comm} no longer apply unless the relevant population dynamics close on $\mathcal{B}$.

For clarity, we do not model the Clifford gates themselves; indeed, most Cliffords are not $\mathcal{B}$-monomial (e.g., the Hadamard gate). The Clifford assumption serves only to ensure that modeling per-gate noise as stochastic Pauli is consistent with the aggregate noise being stochastic Pauli.

\paragraph{Stationary effective noise.}
For the repeated-cycle and interval analyses below, we assume that the
effective gate-noise channel and QED-check instrument are stationary in
the chosen toggling frame, with kernels $G_\eta$ and $T_\varepsilon$
independent of the cycle index and interval length $m$.
Clifford conjugation preserves the Pauli noise model; this additional
homogeneity assumption gives the repeated-kernel form
$T_\varepsilon(G_\eta)^m$.

\subsection{Noise Matrix}

Assume a gate layer with noise strength $\eta$ and transition matrix $G_{\Lambda'(\eta)}$ (abbreviated $G_\eta$ for fixed Pauli Stochastic channel $\Lambda'$). Since we model gate layers in the toggling frame where ideal logical gates are factored out (see Section~\ref{sec:noise_ass}), $G_{\eta}$ represents the effective Pauli-stochastic noise channel conjugated into a fixed frame. Hence $G_0 = I$ on $B$-populations.

Because ancillas are reset to $\ket{0^a}$ before each QED cycle, any ancilla errors introduced during the gate layer are erased and do not affect subsequent dynamics. The gate layer therefore acts nontrivially only on the data subspace $\mathcal{B}_\mathcal{C}$. Moreover, since the gate layer involves no measurement, it cannot induce transitions to the rejection state. Consequently, $G_\eta$ has the block structure
\begin{equation}
G_\eta = \begin{bmatrix} G_\eta^{(D)} & 0 \\ 0 & 1 \end{bmatrix},
\label{eq:G_eta_block}
\end{equation}
where $G_\eta^{(D)}$ is a $2^n \times 2^n$ column-stochastic matrix on $\mathcal{B}_\mathcal{C}$, and the second block corresponds to the absorbing rejection state $R$.

\subsection{QED Clock Cycle Matrix}

We now construct a transition matrix describing a \emph{QED clock cycle}: a sequence of $m$ logical gates followed by a stabilizer check with post-selection.

\begin{definition}[QED Clock Cycle Matrix]
    Consider a stabilizer QED code with noise strength $\varepsilon$ in the QED cycle and noise strength $\eta$ per logical gate. The QED clock cycle matrix for $m$ logical gates is
    \begin{equation}
        M^{(m)}_{\varepsilon, \eta} := T_\varepsilon (G_\eta)^m.
    \end{equation}
When $m$ is omitted, we assume $m=1$ unless otherwise stated.
\end{definition}

We next express $M^{(m)}_{\varepsilon, \eta}$ as a perturbation from the ideal.

\begin{definition}[Perturbation Matrices]
    Define the perturbations
    \begin{equation}
\Delta^{(T)}_\varepsilon := T_\varepsilon - T_0, \qquad \Delta^{(G,m)}_\eta := (G_\eta)^m - I, \qquad \Delta^{(M,m)}_{\varepsilon,\eta} := M^{(m)}_{\varepsilon,\eta} - T_0.
\label{eq:perturbation_notation}
\end{equation}
\end{definition}

\begin{lemma}[Perturbative Expansion of $M^{(m)}_{\varepsilon,\eta}$]
\label{lem:perturbative_expansion}
For any $m \geq 1$,
    \begin{equation}
   M^{(m)}_{\varepsilon,\eta} = T_0 + \Delta_\varepsilon^{(T)} + T_0\Delta_\eta^{(G,m)} + \Delta_\varepsilon^{(T)}\Delta_\eta^{(G,m)}.
\label{eq:perturbative_expansion}
\end{equation}
\end{lemma}
\begin{proof}
    Substituting the definitions from~\eqref{eq:perturbation_notation} into $M^{(m)}_{\varepsilon,\eta} = T_\varepsilon (G_\eta)^m$ and expanding yields the result.
\end{proof}

\begin{remark}
    In Appendix~\ref{app:clock_cycle_spectrum}, we show that $\|\Delta_\varepsilon^{(T)}\|_1=O(\varepsilon)$ and $\|\Delta_\eta^{(G,m)}\|_1=O(m\eta)$, so $\|\Delta_\varepsilon^{(T)}\Delta_\eta^{(G,m)}\|_1=O(m\varepsilon\eta)$ is second-order in noise strength for constant $m$. Consequently, Equation~\eqref{eq:perturbative_expansion} shows that to first order, the total error decomposes additively: contributions from a noisy QED cycle (with noiseless gates) and from $m$ noisy gates (followed by an ideal QED check) enter independently.
\end{remark}

When using the QED transition matrix as a proof tool, it is often convenient to work directly with the \emph{accepted-space block}, defined by
\begin{equation}
A_{\varepsilon,\eta} := PM_{\varepsilon,\eta}P^T, \qquad P := [I_{L+X} \; 0],
\label{eq:accepted_block_def}
\end{equation}
which encodes the same information as $M_{\varepsilon,\eta}$ while removing the trivial absorbing eigenvalue. Writing $v^{(t)} = [w^{(t)}; r^{(t)}]$, the accepted population evolves as $w^{(t)} = A_{\varepsilon,\eta}^t w^{(0)}$. Spectral properties of the full clock-cycle matrix $M^{(m)}_{\varepsilon,\eta}$ and its accepted block $A^{(m)}_{\varepsilon,\eta}$, including spectrum reduction, diagonalization of $T_0$, and small-noise eigenvalue localization, are collected in Appendix~\ref{app:clock_cycle_spectrum}.

\section{Emergent Logical Non-Markovianity}

\label{sec:nonmarkov}
We now use the framework developed in Sections~\ref{sec:modeling_tm} and~\ref{sec:clock_cycle} to analyze emergent logical non-Markovianity.
Under memoryless physical noise, one might expect logical error statistics to also be memoryless. Concurrent work~\cite{ziyad2025emergent} shows this expectation fails: conditioned logical statistics can exhibit history dependence even when physical operations are Markovian. This can have notable impacts on characterizing logical qubits; characterization techniques that do not take this behavior into account can fail to capture steady-state code dynamics \cite{kumar2026co}. Building on this viewpoint, we use the explicit accepted-space kernel $A$ to isolate and quantify a leading-order, experimentally visible transient in acceptance-conditioned logical diagnostics.

The transition-matrix framework provides a clean picture. The dynamics on the expanded state space (logical $\cup$ leakage $\cup$ rejection) \emph{is} Markovian: it evolves by left-multiplication by a fixed stochastic matrix. Accepted syndrome populations carry information between cycles and influence subsequent acceptance and fidelity. Writing $w^{(t)}$ for the subnormalized accepted population, the conditional distribution $\widehat{w}^{(t)} := w^{(t)} / \|w^{(t)}\|_1$ evolves as
\begin{equation}
    \widehat{w}^{(t+1)} = \frac{A \, \widehat{w}^{(t)}}{\|A \, \widehat{w}^{(t)}\|_1}.
\end{equation}
The denominator depends on the current state, making this a nonlinear map.
We use a gate-composability notion of logical Markovianity,
following~\cite{ziyad2025emergent}: each application of the same logical
operation is represented by the same logical map. For post-selected QED,
we test this requirement using a single near-identity, completely positive,
trace-nonincreasing accepted map acting only on the logical code space.

\subsection{Accepted-State Correctness}
To characterize this phenomenon using the transition matrix, we use the
accepted-state correctness (ASC) diagnostic of~\cite{kumar2026co}, i.e.\
the fidelity conditioned on the QED checks passing.
Fix a target logical basis state $\ket{\psi}_L\in\mathcal{B}_\mathcal{C}$,
let $e_\psi$ denote its coordinate vector in the accepted space, and
initialize $w^{(0)}=e_\psi$.
Throughout, we use the same symbol for its restriction to the logical
sector or zero extension to the full reduced space, as required by the
block dimensions.

\begin{definition}[Accepted-State Correctness]
For $t$ such that $p_{\mathrm{acc}}(t) > 0$, define
\begin{equation}
    \mathrm{ASC}(t) := \frac{e_\psi^T A^t e_\psi}{\bigl\|A^t e_\psi\bigr\|_1},
\end{equation}
where $A$ is the accepted block of the QED transition matrix.
\end{definition}

\noindent Because the QED cycle is $\mathcal{B}$-nonactivating, the transition matrix $A$ gives the exact evolution of the relevant diagonal elements when the initial state is diagonal. Fidelity to any basis state is the corresponding diagonal element, so
\begin{equation}
    \mathrm{ASC}(t) = F\!\left(\ket{\psi}_L\!\bra{\psi}_L, \, \widehat{\rho}^{(t)}_{\mathrm{acc}}\right),
\end{equation}
Where $F(\rho, \ket{\psi}\bra{\psi}) = \braket{\psi|\rho|\psi}$ is the fidelity of $\rho$ with respect to $|\psi\rangle$ and $\widehat{\rho}^{(t)}_{\mathrm{acc}}$ is the normalized post-selected state after $t$ cycles.

\subsection{First-Cycle Transient in ASC}

We now show that, to leading order in noise strength, ASC exhibits logical non-Markovianity near the beginning of the code evolution. At $t=0$, the system is perfectly in the code space. After one noisy cycle, some population leaks to non-code-space sectors while remaining accepted (a ``false negative''). Subsequent cycles reject this leakage while adding fresh leakage from their own noisy QED checks, producing a deviation from naive memoryless predictions.

Let $A_0 = \mathrm{diag}(I_L, 0_X)$ be the accepted block of the noiseless QED matrix, and let $\Delta A := A - A_0$ with $\|\Delta A\|_1 = O(\varepsilon+\eta)$.

\begin{proposition}[First-Order Expansion of ASC]
\label{prop:ASC_expansion}
For any fixed $t \geq 1$,
\begin{equation}
    \mathrm{ASC}(t) = 1 + t(e_\psi^T \Delta A \, e_\psi) - (\mathbf{1}_{L+X}^T \Delta A \, e_\psi) - (t-1)(\mathbf{1}_L^T (\Delta A)_{LL} e_\psi) + O(\|\Delta A\|_1^2).
\label{eq:ASC_expansion}
\end{equation}
\end{proposition}

\begin{proof}
We expand $\mathrm{ASC}(t) = (e_\psi^T A^t e_\psi)/(\bigl\|A^t e_\psi\bigr\|_1)$ to first order in $\Delta A$. Writing $A = A_0 + \Delta A$ and expanding $(A_0 + \Delta A)^t$, the first-order terms sum to $\sum_{j=0}^{t-1} A_0^{t-1-j} \Delta A \, A_0^j$. Applying to $e_\psi$ and using $A_0^j e_\psi = e_\psi$ for all $j \geq 0$ yields $\sum_{j=0}^{t-1} A_0^{t-1-j} \Delta A \, e_\psi$. Since $A_0^2 = A_0$, the $j = t-1$ term contributes $\Delta A \, e_\psi$ while the remaining $t-1$ terms each contribute $A_0 \Delta A \, e_\psi$. Thus
\begin{equation}
    A^t e_\psi = e_\psi + \Delta A \, e_\psi + (t-1) A_0 \Delta A \, e_\psi + O(\|\Delta A\|_1^2).
\label{eq:At_expansion}
\end{equation}
For the numerator of ASC, apply $e_\psi^T$ to~\eqref{eq:At_expansion}. Using $e_\psi^T A_0 = e_\psi^T$ (since $A_0$ acts as identity on the logical sector containing $e_\psi$):
\begin{equation}
    e_\psi^T A^t e_\psi = 1 + (e_\psi^T \Delta A \, e_\psi) + (t-1)(e_\psi^T \Delta A \, e_\psi) + O(\|\Delta A\|_1^2) = 1 + t(e_\psi^T \Delta A \, e_\psi) + O(\|\Delta A\|_1^2).
\end{equation}
For the denominator, apply $\mathbf{1}_{L+X}^T$ to~\eqref{eq:At_expansion}. Using $\mathbf{1}_{L+X}^T A_0 = (\mathbf{1}_L^T, 0_X)$ (since $A_0$ annihilates the leakage sector):
\begin{equation}
    \bigl\|A^t e_\psi\bigr\|_1 = 1 + (\mathbf{1}_{L+X}^T \Delta A \, e_\psi) + (t-1)(\mathbf{1}_L^T (\Delta A)_{LL} e_\psi) + O(\|\Delta A\|_1^2).
\end{equation}
Taking the ratio and expanding to first order yields~\eqref{eq:ASC_expansion}.
\end{proof}

The inner products appearing in~\eqref{eq:ASC_expansion} have natural interpretations. For convenience, define
\begin{align}
    a &:= e_\psi^T (\Delta A) e_\psi, &
    b_{\mathrm{acc}} &:= \mathbf{1}_{L+X}^T (\Delta A) e_\psi, &
    b_{\mathrm{log}} &:= \mathbf{1}_L^T (\Delta A)_{LL} e_\psi,
\end{align}
so that $\mathrm{ASC}(t) = 1 + ta - b_{\mathrm{acc}} - (t-1)b_{\mathrm{log}} + O(\|\Delta A\|_1^2)$. Here $a$ is the first-order change in the correct-logical-state probability, $b_{\mathrm{acc}}$ is the first-order change in total accepted probability, and $b_{\mathrm{log}}$ is the corresponding change restricted to the logical sector. The structure of~\eqref{eq:ASC_expansion} shows that $a$ accumulates every cycle, while $b_{\mathrm{log}}$ begins accumulating only after the first cycle. The difference $\ell := b_{\mathrm{acc}} - b_{\mathrm{log}} = \bigl\|(\Delta A)_{XL} e_\psi\bigr\|_1 \geq 0$ is the \emph{accepted leakage injection}: the probability that a run starting in the correct logical state lands in an accepted leakage sector after one cycle.

\begin{corollary}[Per-Cycle Drift]
\label{cor:per_cycle_drift}
Define $h(t) := \mathrm{ASC}(t+1) - \mathrm{ASC}(t)$. Then
\begin{equation}
    h(0) = a - b_{\mathrm{acc}} + O(\|\Delta A\|_1^2), \qquad h(t) = a - b_{\mathrm{log}} + O(\|\Delta A\|_1^2) \;\; \text{for } t \geq 1,
\end{equation}
and consequently $h(0) - h(1) = -\ell + O(\|\Delta A\|_1^2)$.
\end{corollary}

Since $\ell \geq 0$, the initial per-cycle drift is generically more negative than subsequent drifts.

\paragraph{Obstruction to a stationary logical-only model.}
For this comparison, take $\eta=r\varepsilon$ with fixed $r\geq0$
as $\varepsilon\to0$.
Let $\Phi_\varepsilon=\mathrm{id}+\varepsilon\mathcal{K}
+O(\varepsilon^2)$ be a candidate logical-only accepted map, initialized at
$\rho_\psi=\ket{\psi}_L\!\bra{\psi}_L$.
For fixed $t$, $\Phi_\varepsilon^t=\mathrm{id}+t\varepsilon\mathcal{K}
+O(\varepsilon^2)$, so its acceptance probability and acceptance-conditioned
fidelity to $\rho_\psi$ satisfy
\[
p_\Phi(t)=1+t\varepsilon r_\psi+O(\varepsilon^2),
\qquad
F_\Phi(t)=1+t\varepsilon f_\psi+O(\varepsilon^2),
\]
with $r_\psi$ and $f_\psi$ independent of $t$.
Thus each diagnostic has cycle-independent increments to first order.
By contrast, the denominator expansion above, with
$p_{\mathrm{acc}}(t):=\|A^t e_\psi\|_1$, gives
\[
p_{\mathrm{acc}}(2)-2p_{\mathrm{acc}}(1)+1
=-\ell+O(\varepsilon^2),
\]
and Corollary~\ref{cor:per_cycle_drift} gives
$h(0)-h(1)=-\ell+O(\varepsilon^2)$.
Whenever $\ell=\lambda\varepsilon+O(\varepsilon^2)$ with $\lambda>0$,
no such stationary logical-only map reproduces these diagnostics through
first order. This is a leading-order signature of logical non-Markovianity
in the gate-composability sense specified above.
The same accepted-leakage injection controls both transients; in particular,
$\ell=2p_{\mathrm{acc}}(1)-p_{\mathrm{acc}}(2)-1+O(\varepsilon^2)$
can be inferred from one- and two-cycle acceptance probabilities, without
logical-state readout.

We now identify the physical origin of $\ell$. From the perturbative expansion~\eqref{eq:perturbative_expansion}, taking accepted blocks gives $\Delta A = (\Delta_\varepsilon^{(T)})_{\mathrm{acc}} + (T_0 \Delta_\eta^{(G)})_{\mathrm{acc}} + O(\varepsilon\eta)$. Since $(T_0 \Delta_\eta^{(G)})_{XL} = 0$ (gate-induced leakage is immediately rejected by the ideal check), only QED-check noise contributes to $\ell$ at first order:
\begin{equation}
    \ell = \bigl\|((\Delta_\varepsilon^{(T)})_{\mathrm{acc}})_{XL} e_\psi\bigr\|_1 + O((\varepsilon + \eta)^2).
\end{equation}
Physically, QED-check imperfections produce \emph{false negatives}: the data is driven into a nontrivial-syndrome sector while the ancilla still returns $0^a$. These accepted leakage states persist briefly before being rejected, producing the first-cycle transient in conditional metrics like ASC.

\section{Optimizing QED Efficiency via the QED Interval \label{sec:qed_efficiency}}
In this section, we analyze the \emph{QED efficiency}, a recently proposed metric that quantifies the additive error reduction per shot due to a QED code \cite{kumar2026co}. For completeness of this manuscript, we re-derive the motivation for this metric in Appendix~\ref{app:se_efficiency_motivation}, originally demonstrated in \cite{kumar2026co}. Beyond the specific results, this section illustrates a broader benefit of the transition matrix formalism: perturbative expansion of stochastic matrices naturally yields closed-form parameters that are physically meaningful and useful for analysis (see Remark~\ref{rem:abc}).

\subsection{Definitions}

\begin{definition}[$m$-unit]
    For a QED code, an $m$-unit is a sequence of $m$ logical gates terminating in a stabilizer check and post-selection.
\end{definition}

\begin{assumption}
In this section, we standardize the structure of the $m$-units we consider.
We consider repeated applications of a fixed transversal Clifford gate with
a stationary effective Pauli-noise layer $G_\eta$, as in
Section~\ref{sec:noise_ass}. We assume that a weight-$w$ fault pattern
induces an effective Pauli of weight at most $c_{\text{prop}}w$,
with $c_{\text{prop}}$ independent of $m$.
A sequence of $m$ logical gates is therefore described by $(G_\eta)^m$.
Finally, we consider $m$-units run \textit{after} the initial QED transient
described in Section~\ref{sec:nonmarkov}, to characterize post-transient
behavior to first order.
\label{ass:eff_gate}
\end{assumption}

\begin{assumption}[Unified noise parameter]
Since the QED check and the logical gates operate on the same physical hardware, in this section we assume a fixed common noise strength $\varepsilon = \eta$. This models a realistic scenario where the hardware noise rate is given, and one wishes to determine the optimal QED check frequency. Notationally, in this section we drop the second noise parameter $\eta$ and write all noise strengths as $\varepsilon$.

\label{ass:unified_noise}
\end{assumption}

\begin{remark}
Nonidentical effective Pauli layers can be handled by replacing $(G_\eta)^m$
with their ordered product; for a fixed effective check, the first-order
gate contribution $mT_0G_1$ becomes $T_0\sum_{j=1}^m G_{1,j}$,
where $G_{1,j}$ is the first-order coefficient of layer $j$.
Distinct check and gate noise strengths can likewise be retained.
We use the homogeneous setting to obtain the closed-form interval
dependence below.
\end{remark}

\begin{remark}[Bookkeeping in $m$]
All asymptotics in this section are taken as $\varepsilon\to 0$ with $m$ a fixed integer.
When useful, we keep track of how the first neglected Taylor coefficients depend on $m$; because an
$m$-gate interval contains $O(m)$ fault locations, the order-$\varepsilon^2$ coefficients scale as
$O(m^2)$ under circuit-level Pauli-stochastic noise.
\end{remark}

As in Section~\ref{sec:nonmarkov}, fix a target logical basis state $\ket{\psi}_L \in \mathcal{B}_\mathcal{C}$ with basis vector $e_\psi$.

\begin{definition}[Error and rejection rate of an $m$-unit]
We now wish to define the error rate and rejection rate of an $m$-unit. Under Assumptions~\ref{ass:eff_gate} and~\ref{ass:unified_noise}, these depend only on $m$. To consider $m$-units that are run \emph{after} the first-cycle QED transient, we note that the transient appears only in the first cycle to first order, and so we define
\begin{align}
    p_{\text{rej}}^{(m)} &:= 1-\frac{p_\text{acc}^{(m)}(2)}{p_\text{acc}^{(m)}(1)}, & \varepsilon_\ell^{(m)} := 1-\frac{\text{ASC}^{(m)}(2)}{\text{ASC}^{(m)}(1)}
\end{align}
where $p_\text{acc}^{(m)}(t):= \bigl\|(A^{(m)})^t e_\psi\bigr\|_1$ is the probability of acceptance after $t$ applications of an $m$-unit, and $A^{(m)}$ is the accepted block of $M^{(m)}$. The choice of taking the ratio of the second cycle to the first is arbitrary to first order.
\label{def:metrics_m_unit}
\end{definition}
\noindent Here $\varepsilon_\ell^{(m)}$ is an operational logical error rate
extracted from target-basis-state fidelity decay after removal of the
first-order QED transient. To first order, it counts transitions to other
logical basis labels, as shown in Lemma~\ref{lem:se_two_cycle}.

\begin{definition}[Physical error rate for an $m$-unit]
\label{def:physical_error_rate}
For the physical baseline, consider a reference sequence of $m$
physical gates without QED checks, each with independent
Pauli-stochastic fault probability $\varepsilon$.
We use the probability of at least one fault:
\begin{align}
    \varepsilon_p^{(m)} &:= 1 - (1-\varepsilon)^m = m\varepsilon + O(m^2\varepsilon^2).
\end{align}
\end{definition}

\begin{definition}[QED Efficiency]
For $p_{\text{rej}}^{(m)} >0$, the \textit{QED efficiency} of an $m$-unit is
\begin{equation}
        \text{SE}(m) = \frac{\varepsilon_p^{(m)}-\varepsilon_\ell^{(m)}}{p_{\text{rej}}^{(m)}}
    \label{eq:se}
\end{equation}
\end{definition}
Equation~\eqref{eq:se} compares the physical fault-probability baseline
with the operational logical error rate per discarded shot.
Its interpretation in terms of PEC sampling savings~\cite{kumar2026co}
requires these rates to agree, to the retained order, with the
effective-channel parameters under the shot-cost assumptions of
Appendix~\ref{app:se_efficiency_motivation}.

In this section, we use our transition matrix formalism to derive a first-order expansion for $\mathrm{SE}(m)$ and show that, to leading order, increasing $m$ increases $\mathrm{SE}(m)$ (Section~\ref{sec:se_firstord}). This provides a concrete leading-order explanation for empirical trends observed in simulations~\cite{kumar2026co}.

\subsection{QED Efficiency to First Order \label{sec:se_firstord}}

The first-order behavior of $\mathrm{SE}(m)$ is controlled entirely by the first-order perturbation
of the accepted block $A^{(m)}$ of the length-$m$ unit.

\begin{lemma}[Linearization and affine $m$-dependence of $A^{(m)}$]
\label{lem:se_A1_affine}
Under circuit-level stochastic Pauli noise there exist matrices $T_1$ and $G_1$ (independent of
$\varepsilon$) such that, for each fixed $m\ge 1$,
\begin{equation}
A^{(m)}_\varepsilon
=
A_0 + \varepsilon A_1^{(m)} + O(m^2\varepsilon^2),
\qquad
A_1^{(m)} = (T_1)_{\mathrm{acc}} + m\,(T_0G_1)_{\mathrm{acc}},
\label{eq:se_A1_affine}
\end{equation}
where $A_0=\mathrm{diag}(I_L,0_X)$ is the accepted block of the ideal check $T_0$.
\end{lemma}

\begin{proof}[Proof sketch]
Both $T_\varepsilon$ and $G_\varepsilon$ are differentiable around $\varepsilon=0$ due to our assumption of Pauli Stochastic noise at independent fault locations. Expand $T_\varepsilon=T_0+\varepsilon T_1+O(\varepsilon^2)$ and $G_\varepsilon=I+\varepsilon G_1+O(\varepsilon^2)$, then
differentiate $M^{(m)}_\varepsilon=T_\varepsilon(G_\varepsilon)^m$ at $\varepsilon=0$:
$\left.\frac{d}{d\varepsilon}M^{(m)}_\varepsilon\right|_{\varepsilon=0}=T_1+mT_0G_1$. Taking accepted blocks
is linear, yielding \eqref{eq:se_A1_affine}. A fully detailed derivation is in Appendix~\ref{app:se_efficiency_details}.
\end{proof}

\begin{lemma}[Two-cycle extraction from $A_1$]
\label{lem:se_two_cycle}
Let $A_\varepsilon=A_0+\varepsilon A_1+O(\varepsilon^2)$ with $A_0=\mathrm{diag}(I_L,0_X)$ and initial logical
basis state $e_\psi$. With $p_{\mathrm{rej}}$ and $\varepsilon_\ell$ defined as in
Definition~\ref{def:metrics_m_unit}, we have
\begin{align}
p_{\mathrm{rej}}
&= -\varepsilon\,\mathbf{1}_L^T (A_1)_{LL} e_\psi + O(\varepsilon^2),
\label{eq:se_prej_from_A1}\\
\varepsilon_\ell
&= \varepsilon\Big(\mathbf{1}_L^T (A_1)_{LL} e_\psi - e_\psi^T (A_1)_{LL} e_\psi\Big) + O(\varepsilon^2)
= \varepsilon\sum_{\substack{i\in L\\ i\neq \psi}} (A_1)_{i,\psi} + O(\varepsilon^2).
\label{eq:se_ell_from_A1}
\end{align}
\end{lemma}
\begin{proof}[Proof sketch]
Use $A_0^2=A_0$ to expand $A_\varepsilon^2=A_0+\varepsilon(A_0A_1+A_1A_0)+O(\varepsilon^2)$, then expand the
ratios $p_{\mathrm{acc}}(2)/p_{\mathrm{acc}}(1)$ and $\mathrm{ASC}(2)/\mathrm{ASC}(1)$ to first
order via $\frac{1+\varepsilon u}{1+\varepsilon v}=1+\varepsilon(u-v)+O(\varepsilon^2)$. Routine algebra is in
Appendix~\ref{app:se_efficiency_details}.
\end{proof}

\noindent\textbf{Specialization to $m$-units.}
When applying Lemma~\ref{lem:se_two_cycle} to an $m$-unit, the matrix $A_1$ in
\eqref{eq:se_ell_from_A1} should be read as the first-order coefficient $A_1^{(m)}$ from
Lemma~\ref{lem:se_A1_affine}.

\begin{proposition}[First-order forms of $p_{\mathrm{rej}}^{(m)}$ and $\varepsilon_\ell^{(m)}$]
\label{prop:se_first_order_coeffs}
Assume $c_{\mathrm{prop}}<d$ (Assumption~\ref{ass:eff_gate}). For each fixed $m$,
\begin{align}
p_{\mathrm{rej}}^{(m)} &= \varepsilon(\alpha+\beta m)+O(m^2\varepsilon^2),
\label{eq:se_prej_affine}\\
\varepsilon_\ell^{(m)} &= \gamma\,\varepsilon+O(m^2\varepsilon^2),
\label{eq:se_ell_linear}
\end{align}
where
\begin{equation}
\alpha:=-\mathbf{1}_L^T\!\left(((T_1)_{\mathrm{acc}})_{LL}e_\psi\right),\qquad
\beta:=-\mathbf{1}_L^T\!\left(((T_0G_1)_{\mathrm{acc}})_{LL}e_\psi\right),\qquad
\gamma:=
\sum_{\substack{i\in L\\ i\neq \psi}}\bigl(A_1^{(m)}\bigr)_{i,\psi}
=
\sum_{\substack{i\in L\\ i\neq \psi}}\left((T_1)_{\mathrm{acc}}\right)_{i,\psi},
\label{eq:se_alpha_beta_gamma}
\end{equation}
The second equality uses $A_1^{(m)}=(T_1)_{\mathrm{acc}}+m(T_0G_1)_{\mathrm{acc}}$ and the fact that
$(T_0G_1)_{\mathrm{acc}}$ has no off-diagonal logical injections when $c_{\mathrm{prop}}<d$, and
$\alpha,\beta,\gamma\ge 0$.
\end{proposition}

\begin{proof}[Proof sketch]
Apply Lemma~\ref{lem:se_two_cycle} to $A^{(m)}_\varepsilon$ and substitute the affine form
\eqref{eq:se_A1_affine} to obtain \eqref{eq:se_prej_affine}. For $\varepsilon_\ell^{(m)}$, the only
potential $m$-dependence at $O(\varepsilon)$ comes from $(T_0G_1)_{\mathrm{acc}}$, but a single gate
fault propagates to an effective Pauli of weight $\le c_{\mathrm{prop}}<d$, hence it is either
rejected by the ideal check $T_0$ (nontrivial syndrome) or is stabilizer-trivial on the code space
(trivial syndrome). Thus $(T_0G_1)_{\mathrm{acc}}$ has no off-diagonal logical injections, and the
sum in \eqref{eq:se_ell_from_A1} depends only on $(T_1)_{\mathrm{acc}}$, giving
\eqref{eq:se_ell_linear}. Nonnegativity follows from stochasticity and the fact that the relevant
probabilities are $\le 1$ and equal $1$ at $\varepsilon=0$. Full details are in
Appendix~\ref{app:se_efficiency_details}.
\end{proof}

\begin{remark}[Interpretation of $\alpha$, $\beta$, $\gamma$]
    The parameters in~\eqref{eq:se_alpha_beta_gamma} admit direct physical interpretations: $\alpha$ is the fixed first-order check contribution to the post-transient rejection rate, including delayed rejection of accepted leakage, $\beta$ is the first-order rejection coefficient per gate layer followed by an ideal check (measuring detectability of gate errors), and $\gamma$ is the first-order coefficient for accepted transitions to other logical basis labels caused by syndrome-extraction faults.
These coefficients depend on the code and the chosen gate-layer and syndrome-extraction implementations.
\label{rem:abc}
\end{remark}

\begin{theorem}[First-order QED efficiency and monotonicity]
\label{thm:se_first_order}
Assume $\alpha+\beta\gamma>0$ (which, together with $\alpha,\beta,\gamma\ge 0$, gives $\alpha+\beta m>0$ for all $m\ge 1$). For each fixed $m$,
\begin{equation}
\mathrm{SE}(m)=\frac{m-\gamma}{\alpha+\beta m}
+O\!\left(\varepsilon\,\frac{m^2}{\alpha+\beta m}\left(1+\frac{m}{\alpha+\beta m}\right)\right).
\label{eq:se_first_order}
\end{equation}
\noindent
The correction term grows with $m$; when $\alpha+\beta m=\Theta(m)$ (e.g.\ $\beta>0$), it scales as
$O(m\varepsilon)$, so for fixed $\varepsilon$ the first-order approximation can fail at large $m$.
Moreover, the leading term $\mathrm{SE}_0(m):=\frac{m-\gamma}{\alpha+\beta m}$ is strictly
increasing in $m$.
Consequently, for every fixed integer $m_{\max}\ge2$, there exists
$\varepsilon_0>0$ such that the exact efficiency satisfies
$\mathrm{SE}(m+1)>\mathrm{SE}(m)$ for all
$1\le m<m_{\max}$ and $0<\varepsilon<\varepsilon_0$.
\end{theorem}

\begin{proof}[Proof sketch]
By hypothesis $\alpha+\beta\gamma>0$, and since $\alpha,\beta,\gamma\ge 0$ this gives $\alpha+\beta m>0$ for all $m\ge 1$. Combine $\varepsilon_p^{(m)}=m\varepsilon+O(m^2\varepsilon^2)$ with
\eqref{eq:se_prej_affine}--\eqref{eq:se_ell_linear} and cancel a common factor of $\varepsilon$ in
\eqref{eq:se}. A first-order ratio expansion yields the displayed remainder; see
Appendix~\ref{app:se_first_order_remainder}. For monotonicity, compute the forward difference:
\[
\mathrm{SE}_0(m+1)-\mathrm{SE}_0(m)
=
\frac{m+1-\gamma}{\alpha+\beta(m+1)}-\frac{m-\gamma}{\alpha+\beta m}
=
\frac{\alpha+\beta\gamma}{(\alpha+\beta m)\,(\alpha+\beta(m+1))}>0.
\]
For each fixed $m$, the exact forward difference converges to the strictly
positive leading-order difference above. Taking the minimum of the finitely
many corresponding noise thresholds proves the final claim.
\end{proof}
\begin{remark}[Validity of first-order approximation]
The remainder term in~\eqref{eq:se_first_order} comes from multi-fault events within an $m$-unit.
Under circuit-level Pauli-stochastic noise, an $m$-unit contains $O(m)$ fault locations, so
the probability of two or more faults is $O(m^2\varepsilon^2)$, yielding the correction in
\eqref{eq:se_first_order}. The explicit remainder quantifies the accuracy
of the leading-order approximation; when $\beta>0$, its stated scale is
$O(m\varepsilon)$. Accuracy of the efficiency value should be distinguished
from resolution of the smaller differences between neighboring intervals.
Small per-unit rejection probability,
\[
p_{\mathrm{rej}}^{(m)}=\varepsilon(\alpha+\beta m)+O(m^2\varepsilon^2)\ll 1,
\]
keeps the sampling overhead manageable, since the raw-shot overhead per
accepted $m$-unit scales as $1/(1-p_{\mathrm{rej}}^{(m)})$.
\end{remark}

\begin{remark}[Interpretation of First-Order Shot Efficiency]
    The numerator $m-\gamma$ compares the physical fault-probability coefficient $m$ with the accepted logical-label transition coefficient $\gamma$, while the denominator $\alpha+\beta m$ captures the leading-order rejection cost. This illustrates a key benefit of the transition-matrix framework: the ability to write complex quantities in terms of interpretable stochastic parameters. 
\end{remark}

\section{Symmetry-Based Model Reduction}
\label{sec:symmetry_model_reduction}

The reduced QED and clock-cycle kernels developed above act on the $(2^n+1)$-state space
$\overline{\mathcal{Q}}=\mathcal{B}_{\mathcal{C}}\cup\{R\}$ and are \emph{exact} under circuit-level
stochastic Pauli noise. For larger codes, or when a more compact representation is desired (e.g.\ for
analytical or experimental studies), we would like to further coarse-grain $\overline{\mathcal{Q}}$
\emph{without changing} the predicted accepted-shot dynamics for symmetry-compatible observables.

A standard tool is Markov-chain lumping: if a partition of states is (strongly) lumpable for
a kernel $M$, then the coarse population on partition blocks evolves autonomously under an induced
quotient kernel. One source of exact lumpings comes from \emph{symmetries}:
if a group acts by relabelings that leave $M$ unchanged, then the orbit partition yields an exact
quotient chain~\cite{kemeny1976finite}. This reduces an exact $(2^n+1)$-state
model down to a model whose dimension is the number of group orbits (Examples~\ref{ex:422_quotient}
and Appendix~\ref{app:steane_orbit_lumping}). However, the existence of such symmetries depends on both the code being modeled and the noise on the circuits.

One natural source of symmetries in stabilizer codes is wire-permutation symmetries of the data qubits. However, in general, the symmetries of the code being modeled do
\emph{not} automatically appear in the induced transition matrix: they persist only if the noise
model does not distinguish wires. In this section we therefore focus on a simple sufficient noise
structure (Assumption~\ref{assump:ancilla_clean_checks}) under which all implemented wire-permutation
symmetries of the code are preserved. Under these conditions, the transition matrix reduces exactly to the orbits of the implemented symmetry group $G$.

\begin{definition}[Permutation symmetries of a transition matrix]
\label{def:matrix_symmetry}
Let $\Omega$ be a finite state set and let $M$ be a nonnegative matrix indexed by $\Omega$.
For a permutation $\pi\in S_\Omega$, let $P_\pi$ be its permutation matrix.
We say $M$ is \emph{$\pi$-invariant} if
\begin{equation}
M = P_\pi\, M\, P_\pi^{-1}.
\label{eq:matrix_conjugation_symmetry}
\end{equation}
A subgroup $G\le S_\Omega$ is a \emph{symmetry group} of $M$ if \eqref{eq:matrix_conjugation_symmetry}
holds for all $\pi\in G$.
\end{definition}

For a symmetry group $G$, states in the same $G$-orbit have identical coarse transition behavior into every orbit, so the orbit partition yields an exact quotient kernel obtained by summing
transition mass into each orbit (orbit-lumping; see \cite{kemeny1976finite} and
Appendix~\ref{app:symmetry_reduction_details}).

\begin{assumption}[Ancilla-clean checks and homogeneous i.i.d.\ Pauli noise]
\label{assump:ancilla_clean_checks}
For the symmetry-based reductions in this section, we make two additional simplifying assumptions:

\begin{enumerate}
\item \textbf{Homogeneous i.i.d.\ Pauli noise on all data qubits.}
At every fault location on any data qubit, the inserted Pauli is drawn
independently from the \emph{same} single-qubit stochastic Pauli distribution.

\item \textbf{Ancilla-clean check noise.}
The stochastic Pauli faults used to construct the unreduced QED-check matrix $T'_\varepsilon$
act only on \emph{data} qubits (i.e., no ancilla faults during the syndrome-extraction circuit).
\end{enumerate}
\end{assumption}

Throughout this section, an \emph{implemented wire-permutation symmetry}
preserves the data/ancilla partition and admits a relabeling of fault
locations that maps each faulted circuit to its wire-conjugated counterpart,
as in Proposition~\ref{prop:fault_path_covariance}.
Thus the symmetry concerns the fault-location-resolved implementation,
not only the ideal composite operation.

\begin{lemma}[Homogeneous i.i.d.\ Pauli noise preserves wire-permutation symmetries]
\label{lem:iid_preserves_symmetry}
Let $\pi$ be a permutation of the physical qubits (data and ancillas) and let $U_\pi$ be the
corresponding permutation unitary. Suppose $\pi$ is an implemented
wire-permutation symmetry of the circuit component in the sense above.
Under Assumption~\ref{assump:ancilla_clean_checks}, the corresponding \emph{noisy} channel
also has the same symmetry:
\[
\mathcal{E}(U_\pi\rho U_\pi^\dagger)=U_\pi\,\mathcal{E}(\rho)\,U_\pi^\dagger\qquad\forall\rho.
\]
\end{lemma}
\begin{proof}
A wire permutation $\pi$ induces a bijection on Pauli fault patterns by relabeling the qubit indices.
Under homogeneous i.i.d.\ noise (Assumption~\ref{assump:ancilla_clean_checks}), this relabeling preserves
fault probabilities. By the implemented-symmetry hypothesis, the relabeled
fault pattern implements the conjugated faulty circuit. The claim is a direct application of the fault-path
relabeling criterion (Proposition~\ref{prop:fault_path_covariance} in Appendix~\ref{app:symmetry_reduction_details}).
\end{proof}

Combining Lemma~\ref{lem:iid_preserves_symmetry} with the facts that
(i) monomial wire-permutation unitaries induce permutation symmetries of $T_{\mathcal{B}}(\cdot)$ under
covariance, and (ii) those symmetries descend through the QED reduction $\mathcal{R}(\cdot)$,
we conclude that any implemented wire-permutation symmetry group $G$ acting monomially on $\mathcal{B}$ yields an exact
orbit-quotient reduction of $T_\varepsilon$, $G_\eta$, and hence each clock-cycle kernel
$M^{(m)}_{\varepsilon,\eta}=T_\varepsilon(G_\eta)^m$. A formal proof is given in Proposition~\ref{prop:ancilla_clean_symmetry} in Appendix~\ref{app:symmetry_reduction_details}.

\begin{example}[Orbit Quotients of the 4-Qubit Code]
\label{ex:422_quotient}
For the $[[4,2,2]]$ code~\cite{vaidman1996prevention,grassl1997erasure}, the accepted boundary space $\mathcal{B}_{\mathcal{C}}$ has
$2^{n-k}\cdot 2^k = 4\cdot 4=16$ basis labels $(s,\ell)$, and QED reduction adjoins the absorbing
rejection state $R$, giving a $17\times 17$ kernel $M_{\varepsilon,\eta}$.

Under an implemented \emph{data-wire permutation symmetry} $G=S_4$, together with the homogeneous
noise assumptions of Assumption~\ref{assump:ancilla_clean_checks}, the resulting reduced kernel is
$S_4$-invariant and hence admits an \emph{exact} orbit quotient (Appendix Proposition~\ref{prop:ancilla_clean_symmetry}).
These 17 states split into only \emph{seven} $S_4$-orbits: six orbits inside the accepted space plus the
singleton orbit $\{R\}$. Concretely, the accepted space decomposes into the following physically
interpretable macrostates (explicit representatives and a proof are in Appendix~\ref{app:422_example}):

\begin{center}
\begin{tabular}{c c c}
\hline
Orbit & Interpretation & Size \\
\hline
$\mathcal{O}_1$ & target logical basis state & $1$ \\
$\mathcal{O}_2$ & the other three logical basis states & $3$ \\
$\mathcal{O}_3$ & a $Z$-syndrome analogue of $\mathcal{O}_1$ & $1$ \\
$\mathcal{O}_4$ & a $Z$-syndrome analogue of $\mathcal{O}_2$ & $3$ \\
$\mathcal{O}_5$ & the entire $X$-syndrome sector & $4$ \\
$\mathcal{O}_6$ & the entire $Y$-syndrome sector & $4$ \\
$\mathcal{O}_7$ & rejection $R$ (absorbing) & $1$ \\
\hline
\end{tabular}
\end{center}

The resulting $7\times 7$ quotient chain is small enough to inspect directly, yet it still separates
(i) logical confusion ($\mathcal{O}_1\leftrightarrow\mathcal{O}_2$), (ii) accepted leakage sectors
($\mathcal{O}_3$--$\mathcal{O}_6$), and (iii) rejection ($\mathcal{O}_7$). In particular, one can read
off the dominant accepted-to-rejected flow rates and the pathways by which accepted leakage is created
and subsequently rejected, making the model substantially more interpretable than the full 17-state
kernel while remaining exact under the symmetry assumptions.
\end{example}

\begin{remark}[Other symmetric codes]
The $[[4,2,2]]$ code is just a minimal example. Many codes have nontrivial data-wire automorphism
groups and are similarly strong candidates for orbit lumping under homogeneous noise. For example, the
$[[7,1,3]]$ Steane code has a size-$168$ data-wire permutation symmetry group; under the same symmetry
assumptions, orbit lumping reduces the exact $(2^7+1)=129$-state reduced model down to an exact
$11$-state quotient chain (Appendix~\ref{app:steane_orbit_lumping}).
\end{remark}

\begin{remark}
    When the noise channel is strongly biased, some orbits of the reduced transition matrix may be unreachable, enabling further reduction. This is described further in \ref{app:reachability}.
\end{remark}

\begin{remark}
The noise structure in Assumption~\ref{assump:ancilla_clean_checks} is not expected on real hardware.
The point of the homogeneous model here is that it yields an \emph{exact} symmetry and hence an exact
orbit quotient that is small and interpretable. When symmetry holds only approximately,
Appendix~\ref{app:approx_lumping} gives quantitative error bounds for the resulting coarse model.
\end{remark}

\section{Evaluation}
\label{sec:evaluation}

We evaluate the transition-matrix framework along three axes: (i) exactness of the transition matrix under circuit-level Pauli noise (Lemma~\ref{lem:syn_b_comm} and Propositions~\ref{prop:nonact_implies_TM} and~\ref{prop:qed_reduction_exact}); (ii) the first-cycle transient in ASC attributed to QED-check noise (Proposition~\ref{prop:ASC_expansion} and Corollary~\ref{cor:per_cycle_drift}); and (iii) the first-order QED-efficiency formula and its monotonicity in the check interval~$m$ (Theorem~\ref{thm:se_first_order}). Figures~\ref{fig:fig1_exactness}--\ref{fig:fig3_interval} summarize the results.

\subsection{Experimental Setup}
\label{sec:eval_setup}

We use the $[[4,2,2]]$~\cite{vaidman1996prevention,grassl1997erasure} and cyclic perfect $[[5,1,3]]$~\cite{laflamme1996perfect} codes, initialized in the all-zero logical-basis state $e_\psi=e_{0^k}$. The stabilizers are measured in the orders $XXXX$, $ZZZZ$ and $XZZXI$, $IXZZX$, $XIXZZ$, $ZXIXZ$, respectively. Each measurement uses an ideally reset ancilla, with data controls visited in ascending wire order: CNOT for $Z$ factors and $H$--CNOT--$H$ for $X$ factors, with the $H$ gates on the data qubit and the ancilla as the CNOT target.

We adopt a circuit-level \emph{depolarizing} stochastic Pauli noise model: at each single-qubit fault location one of $\{X,Y,Z\}$ is applied with probability $\varepsilon/3$, and at each two-qubit fault location one of the fifteen nontrivial two-qubit Paulis is applied with probability $\varepsilon/15$. Faults are inserted after every Clifford gate in the syndrome-extraction circuit. For gate layers, the ideal transversal Clifford action is removed in the toggling frame of Section~\ref{sec:noise_ass}, so $G_0=I$ on $\mathcal{B}_{\mathcal{C}}$. The repeated-layer model uses stationary effective gate-noise and check kernels as assumed there: $(G_\eta)^m$ retains the noise budget of $m$ logical gate layers without explicitly applying their ideal logical action. Each effective gate-noise layer applies single-qubit depolarization of strength $\eta$ independently to every data qubit. Each syndrome readout independently flips with probability $\varepsilon$.

Figure~\ref{fig:fig1_exactness} compares matrix predictions against Stim~\cite{gidney2021stim} Monte Carlo at $10^6$ shots per code; Figures~\ref{fig:fig2_nonmarkov}--\ref{fig:fig3_interval} are evaluated numerically from the exact accepted-block matrix $A_{\varepsilon,\eta}$. The combined-noise transient series and the efficiency experiments use $\varepsilon=\eta$ as in Assumption~\ref{ass:unified_noise}. In the check-only ablation, $\eta=0$; in the gate-only ablation, $\varepsilon=0$. The horizontal noise-strength parameter in Figure~\ref{fig:fig2_nonmarkov}(b) is applied to the active noise component or components.

These experiments use the full accepted-block construction with noisy ancillas.

\begin{figure}[tbp]
  \centering
  \includegraphics[width=0.9\linewidth]{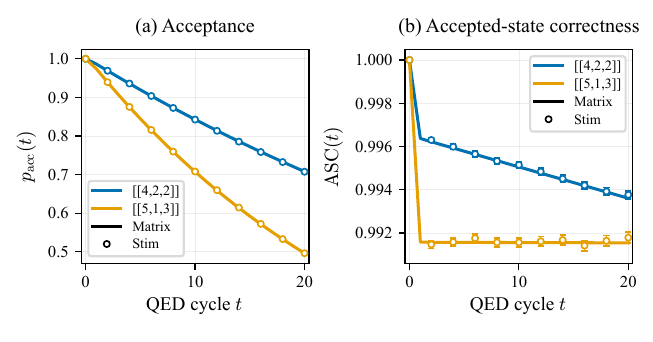}
  \caption{\textbf{Transition-matrix validation.} Lines join exact accepted-block matrix predictions at integer QED cycles; markers show $10^6$-shot Stim estimates (pointwise $95\%$ Wilson confidence intervals) for the $[[4,2,2]]$ (blue) and $[[5,1,3]]$ (orange) codes at $\varepsilon=10^{-3}$ and $\eta=0$. Panel~(a): acceptance probability $p_{\mathrm{acc}}(t)$. Panel~(b): accepted-state correctness $\mathrm{ASC}(t)$. Matrix predictions use Pauli fault averaging and the reduction $\mathcal{R}$ of Section~\ref{sec:qed_lumping}, independently of the simulator. Agreement over $20$ cycles supports the implementation of the population-level construction under the specified circuit and readout noise.}
  \label{fig:fig1_exactness}
\end{figure}

\begin{figure}[tbp]
  \centering
  \includegraphics[width=0.9\linewidth]{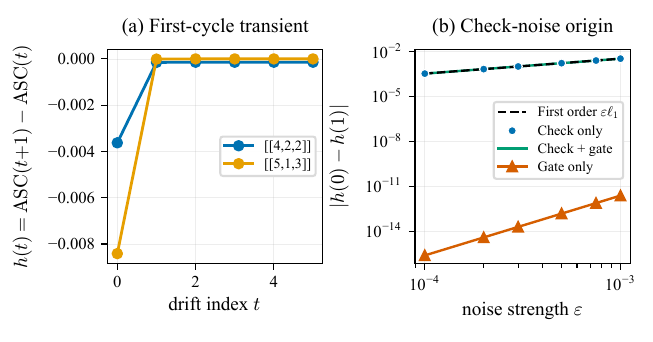}
  \caption{\textbf{Check-noise origin of the first-cycle transient.} Panel~(a): per-cycle ASC drift $h(t)=\mathrm{ASC}(t{+}1)-\mathrm{ASC}(t)$ at $\varepsilon=10^{-3}$ for $[[4,2,2]]$ (blue) and $[[5,1,3]]$ (orange) under check-only noise, showing the characteristic non-zero $h(0)$ followed by an approximately constant drift (Corollary~\ref{cor:per_cycle_drift}). Panel~(b): transient magnitude $|h(0)-h(1)|$ versus $\varepsilon$ for $[[4,2,2]]$, decomposed into check-only, gate-only, and both-noise series. The check-only and combined-noise series nearly overlap the first-order prediction $\varepsilon\,\ell_1$ (black dashed, $\ell_1=3.47$). Panel~(b) uses $60$-digit arithmetic.}
  \label{fig:fig2_nonmarkov}
\end{figure}

\subsection{Results}
\label{sec:eval_results}

\textbf{(i) Transition-matrix validation (Fig.~\ref{fig:fig1_exactness}).}
The transition-matrix predictions $p_{\mathrm{acc}}(t)=\|A^t e_\psi\|_1$ and $\mathrm{ASC}(t)=e_\psi^{T}A^t e_\psi/\|A^t e_\psi\|_1$ agree with $10^6$-shot Stim statistics over $t\in\{0,\ldots,20\}$ for both codes. The largest absolute trajectory discrepancies are approximately $7.9\!\times\!10^{-4}$ for acceptance and $2.4\!\times\!10^{-4}$ for ASC, consistent with finite-shot fluctuations. At $\varepsilon=10^{-3}$ and $t=20$, the matrix predicts $(p_{\mathrm{acc}},\mathrm{ASC})=(0.7078,0.9936)$ for $[[4,2,2]]$ and $(0.4961,0.9915)$ for $[[5,1,3]]$.

\textbf{(ii) Check-noise origin of the first-cycle transient (Fig.~\ref{fig:fig2_nonmarkov}).}
Figure~\ref{fig:fig2_nonmarkov} illustrates Corollary~\ref{cor:per_cycle_drift}: under check-only and combined noise, $|h(0)-h(1)|$ tracks the first-order prediction $\varepsilon\,\ell_1$ with $\ell=\varepsilon\ell_1+O(\varepsilon^2)$ and $\ell_1=3.47$ extracted from $T_1$, while the gate-only series is many orders of magnitude below the check-noise signal. This is consistent with $(T_0\,\Delta_\eta^{(G)})_{XL}=0$, which leaves check noise as the only first-order contributor to accepted leakage injection. The resulting transient illustrates the obstruction to a stationary logical-only accepted-map description specified in Section~\ref{sec:nonmarkov}; the plot does not by itself establish other notions of quantum non-Markovianity. The extracted coefficients are $(\alpha,\beta,\gamma,\ell_1)=(17.47, 4.00, 0.133, 3.47)$ for $[[4,2,2]]$ and $(35.47, 5.00, 0, 8.40)$ for $[[5,1,3]]$. Here $\alpha$ includes delayed rejection of accepted leakage, and $\gamma$ counts first-order transitions to other logical basis labels for the chosen input, as in Remark~\ref{rem:abc}.

\begin{figure}[tbp]
  \centering
  \includegraphics[width=0.9\linewidth]{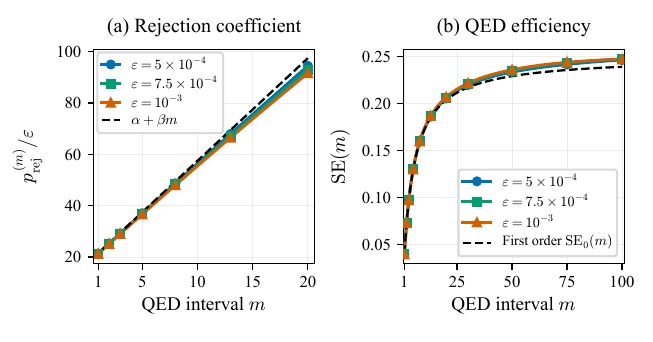}
  \caption{\textbf{First-order QED efficiency on $[[4,2,2]]$.} Panel~(a): rejection coefficient $p_{\mathrm{rej}}^{(m)}/\varepsilon$ versus QED interval $m$ for three noise strengths; the affine prediction $\alpha+\beta m$ (black dashed) approximates the matrix-computed values over $m\in[1,20]$. Panel~(b): QED efficiency $\mathrm{SE}(m)$ versus $m$; matrix-computed values increase over the sampled intervals and track the first-order prediction $\mathrm{SE}_0(m)=(m-\gamma)/(\alpha+\beta m)$ (black dashed) at small $m$. Lines join values at integer intervals.}
  \label{fig:fig3_interval}
\end{figure}

\textbf{(iii) First-order QED efficiency (Fig.~\ref{fig:fig3_interval}).}
Figure~\ref{fig:fig3_interval} evaluates Proposition~\ref{prop:se_first_order_coeffs} and Theorem~\ref{thm:se_first_order} on the primary code $[[4,2,2]]$. Panel~(a) plots $p_{\mathrm{rej}}^{(m)}/\varepsilon$ against $m$ for $\varepsilon\in\{5\!\times\!10^{-4},\, 7.5\!\times\!10^{-4},\, 10^{-3}\}$; the three matrix-computed series approach the affine prediction $\alpha+\beta m$ as the noise strength decreases. Panel~(b) plots $\mathrm{SE}(m)$ versus $m$ over $m\in[1,100]$ together with the first-order prediction $\mathrm{SE}_0(m)=(m-\gamma)/(\alpha+\beta m)$; the computed values increase over the sampled intervals and track $\mathrm{SE}_0$ closely at small $m$, with visible departures at larger intervals consistent with the growing remainder bound in Theorem~\ref{thm:se_first_order}. On $[[4,2,2]]$ at $\varepsilon=10^{-3}$, $\mathrm{SE}(m)$ rises from approximately $0.040$ at $m=1$ to approximately $0.247$ at $m=100$. This increase over the sampled intervals is an observation for the displayed noise settings; the theorem establishes leading-order monotonicity and exact monotonicity on any fixed finite interval range at sufficiently small noise, rather than a uniform guarantee for arbitrary $m$. For the $[[5,1,3]]$ extraction circuit used here, no single check fault produces an accepted transition to a different logical basis label for the chosen input, giving $\gamma=0$. Separately, the $c_{\mathrm{prop}}<d$ hypothesis of Proposition~\ref{prop:se_first_order_coeffs} explains why single gate-layer faults do not contribute first-order off-diagonal logical injections. These constants give an interpretable closed-form surrogate for comparing QED intervals and estimating sampling overhead. The efficiency plotted here uses the physical fault-probability baseline and operational logical error rate of Section~\ref{sec:qed_efficiency}; its interpretation as a PEC shot-cost advantage requires the additional effective-channel identification stated in Appendix~\ref{app:se_efficiency_motivation}.

\section{Discussion \label{sec:discussion}}

\subsection{Constructing the Transition Matrix in Practice}
\label{sec:discussion_practical_tm}
The propagation recipe of Section~\ref{sec:construct_noisy_TM} is implementable whenever gate and check noise are Pauli-stochastic (exact for Clifford circuits with Pauli faults) or well-approximated by Pauli channels in the rotating frame. In that regime, the same approach generates the QED check matrix. Alternatively, the matrix entries can also be estimated empirically from cycle-boundary transition statistics on syndrome/logical labels, with rejection lumped as in Section~\ref{sec:qed_lumping}. This can be particularly accurate for memory experiments, where $\eta=0$ and so gate noise assumptions need not apply. When fault enumeration is too expensive, the low-weight truncation introduced in Section~\ref{sec:construct_noisy_TM} retains $z$th-order accuracy with $O(F^z)$ terms.

\subsection{Extension to Quantum Error Correction}
\label{sec:discussion_qec}

Full QEC differs from QED only in the classical control layer; rather than terminating on nontrivial outcomes, one measures syndromes and applies decoder-dependent corrections (or Pauli-frame updates). At the population level, this decoder step constitutes a conditional stochastic process: for each observed syndrome (and possibly finite-memory summary of past records), the decoder selects a Pauli correction acting as a permutation on $\mathcal{Q}_{\mathcal{C}}$, composing transparently with our cycle matrices. Concretely, augmenting the Markov state with decoder memory (trivial for memoryless decoders; finite for windowed decoders~\cite{skoric2023parallel}) yields a closed Markov process on (data-basis label $\times$ decoder-memory), with one-step kernel obtained by summing syndrome-conditioned branches followed by corresponding correction permutations.

QEC codes are susceptible to the same undetected-leakage effects arising from imperfect syndrome extraction circuits. Therefore, the remark of Section~\ref{sec:nonmarkov} also applies here: QEC induces an exact Markov process on an expanded state
(including syndrome/leakage sectors and any decoder memory), while the logical subsystem is obtained by coarse-graining.
Thus, even under i.i.d.\ circuit-level noise, the effective logical dynamics undergo similar non-Markovian effects, concurring with \cite{ziyad2025emergent}.

\section*{Acknowledgments}

Rohan S. Kumar is supported by the National Science Foundation Graduate Research Fellowship Program (NSF GRFP).

\bibliographystyle{plain}
\bibliography{refs}

\appendix

\section{Extended Discussion of Related Work \label{app:related_works}}

\paragraph{Comparison to Ziyad~\emph{et al.}~\cite{ziyad2025emergent}.}
Ziyad~\emph{et al.} define an emergent-logical-non-Markovianity condition and, in a ``stochastic picture,'' model syndrome-extraction dynamics for Pauli-noise examples as a Markov chain specified by a transition matrix.
Our transition-matrix formalism is compatible with this viewpoint, but we focus on (i) an explicit code-adapted label space indexed by syndrome and logical-eigenvalue labels and (ii) an exact absorbing-state reduction that models multi-cycle \emph{post-selected} QED by iterating a single reduced kernel.
Within this setting we derive closed-form perturbative expressions (e.g.\ accepted leakage injection and the first-cycle transient of accepted-state correctness) that isolate which circuit components contribute at leading order.

\paragraph{Comparison to Kwiatkowski~\emph{et al.}~\cite{kwiatkowski2025constructing}.}
Kwiatkowski~\emph{et al.} work in a Clifford frame with explicit syndrome/outcome registers and prove that, for sufficiently low Pauli-stochastic noise, consecutive QEC cycles admit an approximate \emph{logical-only} Markovian model whose deviations are exponentially suppressed in the number of cycles.
In contrast, we keep an exact finite-state kernel on an expanded space (logical $\cup$ accepted leakage $\cup$ rejection), which is designed to retain finite-time transients and post-selection effects in QED; logical dynamics is then obtained by coarse-graining and conditioning.

\paragraph{Frequency optimization.}
Abu-Nada~\emph{et al.}~\cite{abunada2015optfreq} study how often to apply fault-tolerant QEC operations, noting that correcting after every gate may be suboptimal when the correction procedure is noisy.
Our results analyze an analogous tradeoff for \emph{post-selected QED} using the QED-efficiency metric: we derive the leading-order benefit of longer intervals and track the growth of higher-order corrections.

\section{Resource-destroying-map terminology and diagonal-state persistence}
\label{app:rdm_terminology}

Liu, Hu, and Lloyd introduced resource-destroying maps (RDMs) to express ``resource-free'' properties across resource theories~\cite{liu2017resource}. Given an RDM $\lambda$ (idempotent, fixing precisely free states), they define:
\begin{itemize}
\item \emph{Resource nongenerating}: $\mathcal{E}\circ\lambda=\lambda\circ\mathcal{E}\circ\lambda$
\item \emph{Resource nonactivating}: $\lambda\circ\mathcal{E}=\lambda\circ\mathcal{E}\circ\lambda$
\item \emph{Resource commuting}: $\lambda\circ\mathcal{E}=\mathcal{E}\circ\lambda$
\end{itemize}
In coherence theory, the canonical RDM is complete dephasing~\cite{liu2017resource}. When $\lambda$ is linear, these classes are closed under convex combinations, so averaging over fault patterns preserves the resource-free property.

We adopt the RDM lens with $\lambda$ as dephasing in a basis $B=\{\ket{b_i}\}$, enabling crisp criteria for when transition-matrix descriptions are exact. Our contribution is to (i) exhibit an explicit stabilizer-code-specific basis, (ii) prove stabilizer circuits under stochastic Pauli noise satisfy the corresponding conditions, and (iii) leverage this for practical transition matrices including QED post-selection.

\begin{definition}[$B$-nongenerating~\cite{liu2017resource}]
A completely positive, trace non-increasing map $\mathcal{E}$ is \emph{$B$-nongenerating} iff
\begin{equation}
\mathcal{E}\circ\Delta_B = \Delta_B\circ\mathcal{E}\circ\Delta_B.
\label{eq:Bnong_def}
\end{equation}
In the resource-destroying-map terminology of~\cite{liu2017resource}, $\mathcal{E}$ is $B$-commuting
(Eq.~\eqref{eq:Bcomm_def}) iff it is both $B$-nongenerating~\eqref{eq:Bnong_def} and
$B$-nonactivating~\eqref{eq:Bnonact_def}.
\end{definition}

\begin{remark}[Full state capture on $B$-diagonal ensembles]
If $\mathcal{E}$ is $B$-nongenerating and the initial state is $B$-incoherent
(i.e., $\Delta_B(\rho^{(0)})=\rho^{(0)}$), then every iterate remains $B$-incoherent and the full
density operator equals its diagonal:
\[
\rho^{(t)}=\Delta_B(\rho^{(t)})=\sum_i (p_B(\rho^{(t)}))_i\,\Pi_i.
\]
Thus, within the $B$-incoherent sector, the population vector $p_B(\rho^{(t)})$ determines the full
state. Since $B$-commuting implies $B$-nongenerating, this observation applies to all channels
considered in the main text. In this paper we set $B=\mathcal{B}$, the joint code-adapted basis; $B$-commuting is certified by monomial Kraus representations.
\end{remark}

\section{Proofs for population-level closure and monomiality}
\label{app:prelim_proofs}

\subsection{Proof of Proposition~\ref{prop:nonact_implies_TM}}

\begin{proof}
For any operator $X$, $\Tr(\Pi_i X)=\Tr(\Pi_i\Delta_B(X))$ because $\Delta_B$ preserves the diagonal in basis $B$. Using $B$-nonactivation $\Delta_B\circ\mathcal{E}=\Delta_B\circ\mathcal{E}\circ\Delta_B$:
\begin{align*}
\Tr\!\bigl(\Pi_i\,\mathcal{E}(\rho)\bigr)
&= \Tr\!\bigl(\Pi_i\,\Delta_B(\mathcal{E}(\rho))\bigr)
= \Tr\!\bigl(\Pi_i\,\Delta_B(\mathcal{E}(\Delta_B(\rho)))\bigr) \\
&= \Tr\!\Bigl(\Pi_i\,\Delta_B\Bigl(\mathcal{E}\Bigl(\sum_j \Tr(\Pi_j\rho)\,\Pi_j\Bigr)\Bigr)\Bigr)
= \sum_j \Tr(\Pi_j\rho)\,\Tr\!\bigl(\Pi_i\,\mathcal{E}(\Pi_j)\bigr),
\end{align*}
which equals $(T_B(\mathcal{E})p_B(\rho))_i$ by definition~\eqref{eq:TB_def}. The iterated identity~\eqref{eq:TB_exact_iter} follows by induction on $t$.
\end{proof}

\subsection{Proof of Lemma~\ref{lem:monomial_commutes_dephase}}
\label{app:monomial_commutes}

\begin{lemma}[Restatement of Lemma~\ref{lem:monomial_commutes_dephase}]
Let $\mathcal{E}$ be a CP--TNI map with Kraus operators $\{K_\alpha\}_\alpha$ such that each $K_\alpha$ is monomial in basis $B$. Then $\mathcal{E}$ is $B$-commuting:
\[
\Delta_B\circ \mathcal{E} = \mathcal{E}\circ \Delta_B.
\]
In particular, $\mathcal{E}$ is $B$-nonactivating.
\end{lemma}

\begin{proof}
It suffices to verify the commutation relation on matrix units $\ket{b_i}\!\bra{b_j}$, which span the operator space.

Since $K_\alpha$ is monomial in $B$, there exists a permutation $\pi_\alpha$ of basis indices and nonzero scalars $\{c_{\alpha,i}\}$ such that
\[
K_\alpha\ket{b_i} = c_{\alpha,i}\ket{b_{\pi_\alpha(i)}}.
\]
Applying $K_\alpha$ to a matrix unit:
\[
K_\alpha\ket{b_i}\!\bra{b_j}K_\alpha^\dagger = c_{\alpha,i}c_{\alpha,j}^*\,\ket{b_{\pi_\alpha(i)}}\!\bra{b_{\pi_\alpha(j)}}.
\]

We now verify commutativity for the two cases:

\textbf{Case 1: Diagonal elements ($i=j$).}
We have $K_\alpha\ket{b_i}\!\bra{b_i}K_\alpha^\dagger = |c_{\alpha,i}|^2\ket{b_{\pi_\alpha(i)}}\!\bra{b_{\pi_\alpha(i)}}$, which is again diagonal. Thus:
\[
(\mathcal{E}\circ\Delta_B)(\ket{b_i}\!\bra{b_i}) = \mathcal{E}(\ket{b_i}\!\bra{b_i}) = \sum_\alpha |c_{\alpha,i}|^2\ket{b_{\pi_\alpha(i)}}\!\bra{b_{\pi_\alpha(i)}} = (\Delta_B\circ\mathcal{E})(\ket{b_i}\!\bra{b_i}).
\]

\textbf{Case 2: Off-diagonal elements ($i\neq j$).}
Since $\pi_\alpha$ is a bijection, $\pi_\alpha(i)\neq\pi_\alpha(j)$ whenever $i\neq j$. Thus $K_\alpha\ket{b_i}\!\bra{b_j}K_\alpha^\dagger$ is off-diagonal for each $\alpha$, and
\[
\mathcal{E}(\ket{b_i}\!\bra{b_j}) = \sum_\alpha c_{\alpha,i}c_{\alpha,j}^*\ket{b_{\pi_\alpha(i)}}\!\bra{b_{\pi_\alpha(j)}}
\]
is a sum of off-diagonal terms. Therefore $(\Delta_B\circ\mathcal{E})(\ket{b_i}\!\bra{b_j}) = 0 = (\mathcal{E}\circ\Delta_B)(\ket{b_i}\!\bra{b_j})$.

Combining both cases, $\Delta_B\circ\mathcal{E} = \mathcal{E}\circ\Delta_B$ on all matrix units, hence on all operators by linearity.

Finally, $B$-commuting implies both $B$-nongenerating and $B$-nonactivating by the general theory of resource-destroying maps: if $\Delta_B\circ\mathcal{E} = \mathcal{E}\circ\Delta_B$, then applying $\Delta_B$ before $\mathcal{E}$ to a $B$-diagonal state keeps it $B$-diagonal ($B$-nongenerating), and applying $\Delta_B$ after $\mathcal{E}$ is equivalent to applying it before, so off-diagonal coherences cannot activate to affect diagonal populations ($B$-nonactivating).
\end{proof}

\subsection{Proof of Lemma~\ref{lem:paulis_monomial}}
\label{app:paulis_monomial}

\begin{lemma}[Restatement of Lemma~\ref{lem:paulis_monomial}]
Every Pauli $P\in\mathcal{P}_n$ acts as a signed permutation on the data basis $\mathcal{B}_{\mathcal{C}}$. Every Pauli on the $a$ ancillas acts as a signed permutation on the computational basis $\{\ket{y}\}_{y\in\{0,1\}^a}$. Consequently, every Pauli $P\in\mathcal{P}_{n+a}$ is monomial in the joint basis $\mathcal{B}=\{\ket{s,\ell}\otimes\ket{y}\}$.
\end{lemma}

\begin{proof}
Fix a data basis vector $\ket{s,\ell}=E_s\ket{\ell}_L$ in $\mathcal{B}_{\mathcal{C}}$. Because $P$ and $E_s$ are Paulis, their product satisfies
\[
PE_s = \omega\,E_{s'}L S
\]
for some phase $\omega\in\{\pm 1, \pm i\}$, some syndrome representative $E_{s'}$ (with syndrome $s'=s\oplus s(P)$, where $s(P)$ is the syndrome of $P$), some logical Pauli $L\in N(\mathcal{S})\cap\mathcal{P}_n$, and some stabilizer $S\in\mathcal{S}$.

Since $S\in\mathcal{S}$ fixes the logical state $\ket{\ell}_L$ (i.e., $S\ket{\ell}_L = \ket{\ell}_L$) and $L$ permutes the joint eigenbasis $\{\ket{\ell}_L\}_{\ell\in\{0,1\}^k}$ up to phase (each logical Pauli either fixes or flips each bit of $\ell$), we obtain
\[
P\ket{s,\ell} = PE_s\ket{\ell}_L = \omega\,E_{s'}LS\ket{\ell}_L = \omega'\ket{s',\ell'}
\]
for some phase $\omega'\in\{\pm 1, \pm i\}$ and some $\ell'\in\{0,1\}^k$. Thus $P$ permutes $\mathcal{B}_{\mathcal{C}}$ up to phases, establishing monomiality on the data register.

On ancillas, single-qubit Pauli operators act on the computational basis by either swapping $\ket{0}$ and $\ket{1}$ or adding a phase, so any tensor product of ancilla Paulis permutes $\{\ket{y}\}_{y\in\{0,1\}^a}$ up to phases.

Tensoring the data and ancilla actions: if $P=P_{\mathrm{data}}\otimes P_{\mathrm{anc}}$, then
\[
P(\ket{s,\ell}\otimes\ket{y}) = (P_{\mathrm{data}}\ket{s,\ell})\otimes(P_{\mathrm{anc}}\ket{y}) = \omega''\ket{s',\ell'}\otimes\ket{y'}
\]
for appropriate $s',\ell',y'$ and phase $\omega''$. This establishes monomiality of $P$ in the joint basis $\mathcal{B}$.
\end{proof}

\section{Proof of Proposition~\ref{prop:qed_reduction_exact}: Exact reduced Markov model for post-selected QED}
\label{app:qed_reduction_exact}

\begin{proposition}[Restated]
Let $T'$ be the unreduced (joint) transition matrix for one noisy stabilizer cycle prior to ancilla measurement. Consider the operational QED procedure that, on each cycle, (i) initializes ancillas to $0^a$, (ii) applies the noisy stabilizer circuit, (iii) measures ancillas in the computational basis, and (iv) discards the run unless the outcome is $0^a$. Then the one-step population update on $\overline{\mathcal{Q}}=\mathcal{B}_{\mathcal{C}}\cup\{R\}$ is exactly $v \mapsto \mathcal{R}(T')v$, with $R$ absorbing. Iteration gives $v^{(t)} = (\mathcal{R}(T'))^t v^{(0)}$.
\end{proposition}

\begin{proof}
It suffices to check the update from a basis state at a cycle boundary. If the run is accepted and the data label at the boundary is $q\in\mathcal{B}_{\mathcal{C}}$, then the ancillas are reinitialized to $0^a$, so the joint starting label is $(q,0^a)$. After applying $T'$, the probability to be in joint label $(q',y)$ is $T'[(q',y),(q,0^a)]$.

Post-selection accepts iff $y=0^a$, in which case the next cycle-boundary data label is $q'$; otherwise the trajectory terminates in $R$. Equations~\eqref{eq:reduce_accept_block}--\eqref{eq:reduce_absorb} therefore give the correct one-step transition probabilities on $\overline{\mathcal{Q}}$ from $q$ to $q'$ and from $q$ to $R$.

Finally, \eqref{eq:reduce_absorb} encodes that once rejected, the run remains rejected. Linearity then yields the same update for any population vector $v$, and iteration follows by induction on $t$.
\end{proof}

\begin{remark}[Lumpability intuition]
The QED reduction is an instance of strong lumping for a Markov process with an absorbing terminal class: all joint basis states with ancilla outcome $y\neq 0^a$ are operationally indistinguishable (the run is discarded immediately), so they can be aggregated into a single absorbing state without affecting any future-time accepted statistics.
\end{remark}

\subsection{Accepted-block recursion and reconstruction}
\label{app:accepted_recursion}

Write the reduced QED kernel in absorbing block form
\[
M_{\varepsilon,\eta}^{(m)} = \begin{bmatrix} A_{\varepsilon,\eta}^{(m)} & 0 \\ q_{\varepsilon,\eta}^{(m)T} & 1 \end{bmatrix},
\]
where $A_{\varepsilon,\eta}^{(m)}$ is the accepted block defined by \eqref{eq:accepted_block_def}. The following identities hold for any $m\ge 1$.

\begin{proposition}[Accepted-space recursion]
\label{prop:accepted_recursion}
Let $v^{(t)}=M_{\varepsilon,\eta}^{(m)t} v^{(0)}$ and write $v^{(t)}=[w^{(t)}; r^{(t)}]$ with $w^{(t)}:=Pv^{(t)}$ and $r^{(t)}:=e_R^T v^{(t)}$. Then for all $t\ge 0$,
\[
w^{(t+1)} = A_{\varepsilon,\eta}^{(m)} w^{(t)}, \qquad r^{(t+1)} = r^{(t)} + q_{\varepsilon,\eta}^{(m)T} w^{(t)}.
\]
In particular, $w^{(t)} = \bigl(A_{\varepsilon,\eta}^{(m)}\bigr)^t w^{(0)}$.
\end{proposition}

\begin{proof}
Multiply the block form of $M_{\varepsilon,\eta}^{(m)}$ by $[w^{(t)}; r^{(t)}]$:
\[
\begin{bmatrix} A_{\varepsilon,\eta}^{(m)} & 0 \\ q_{\varepsilon,\eta}^{(m)T} & 1 \end{bmatrix} \begin{bmatrix} w^{(t)} \\ r^{(t)} \end{bmatrix} = \begin{bmatrix} A_{\varepsilon,\eta}^{(m)} w^{(t)} \\ q_{\varepsilon,\eta}^{(m)T} w^{(t)} + r^{(t)} \end{bmatrix}.
\]
Reading off the top and bottom blocks gives the recursion. The closed form follows by induction.
\end{proof}

\begin{corollary}[Acceptance probability]
\label{cor:acceptance_corollary}
Assume $v^{(0)}$ is normalized. Then for all $t\ge 0$, the acceptance probability is $p_{\mathrm{acc}}(t) = 1-r^{(t)} = \mathbf{1}_{L+X}^T \bigl(A_{\varepsilon,\eta}^{(m)}\bigr)^t w^{(0)}$.
\end{corollary}

\begin{proof}
Since $v^{(t)}$ remains normalized and nonnegative under a column-stochastic $M_{\varepsilon,\eta}^{(m)}$, we have $\mathbf{1}_{L+X}^T w^{(t)} = 1-r^{(t)}$ by definition of $r^{(t)}$. Substituting $w^{(t)} = \bigl(A_{\varepsilon,\eta}^{(m)}\bigr)^t w^{(0)}$ from Proposition~\ref{prop:accepted_recursion} completes the proof.
\end{proof}

\begin{proposition}[Reconstruction from the accepted block]
\label{prop:reconstruct_M_from_A}
Let $A_{\varepsilon,\eta}^{(m)}$ be the accepted block of $M_{\varepsilon,\eta}^{(m)}$. Then
\[
q_{\varepsilon,\eta}^{(m)T} = \mathbf{1}_{L+X}^T - \mathbf{1}_{L+X}^T A_{\varepsilon,\eta}^{(m)},
\]
and hence
\[
M_{\varepsilon,\eta}^{(m)} = \begin{bmatrix} A_{\varepsilon,\eta}^{(m)} & 0 \\ \mathbf{1}_{L+X}^T - \mathbf{1}_{L+X}^T A_{\varepsilon,\eta}^{(m)} & 1 \end{bmatrix}.
\]
\end{proposition}

\begin{proof}
For any non-rejection starting state (i.e., any column of $M_{\varepsilon,\eta}^{(m)}$ excluding the $R$ column), column-stochasticity gives
\[
\sum_{i\in \text{non-rej}} (A_{\varepsilon,\eta}^{(m)})_{ij} + (q_{\varepsilon,\eta}^{(m)})_j = 1.
\]
Writing this for all columns simultaneously yields the formula for $q_{\varepsilon,\eta}^{(m)T}$. Substituting into the absorbing block form gives the reconstruction.
\end{proof}

\section{Spectral and perturbation facts for QED clock-cycle matrices}
\label{app:clock_cycle_spectrum}

Fix $m\ge 1$ and write the reduced clock-cycle matrix in absorbing block form
\begin{equation}
M_{\varepsilon,\eta}^{(m)}=\begin{bmatrix}A_{\varepsilon,\eta}^{(m)}&0\\ q_{\varepsilon,\eta}^{(m)T}&1\end{bmatrix},
\label{eq:M_block_form}
\end{equation}
with accepted block $A_{\varepsilon,\eta}^{(m)}$ defined by \eqref{eq:accepted_block_def}. Let $G_\eta^{(m)}:=(G_\eta)^m$ and recall the perturbation notation from \eqref{eq:perturbation_notation}. The results below specialize standard linear-algebraic facts to this structure.

\subsection{Spectrum reduction to the accepted block}
\label{app:spectrum_proofs}

\begin{proposition}[Spectrum reduction to the accepted block]
\label{prop:spectrum_block}
For $M_{\varepsilon,\eta}^{(m)}$ in the absorbing block form~\eqref{eq:M_block_form},
\[
\sigma(M_{\varepsilon,\eta}^{(m)}) = \sigma(A_{\varepsilon,\eta}^{(m)}) \cup \{1\}
\]
as multisets.
\end{proposition}

\begin{proof}
Since $M_{\varepsilon,\eta}^{(m)}$ is block lower-triangular, its characteristic polynomial factorizes:
\[
\det(\lambda I - M_{\varepsilon,\eta}^{(m)})
=\det(\lambda I - A_{\varepsilon,\eta}^{(m)})\cdot(\lambda-1).
\]
Thus the eigenvalues of $M_{\varepsilon,\eta}^{(m)}$ are the eigenvalues of $A_{\varepsilon,\eta}^{(m)}$ together with the absorbing eigenvalue $1$.
\end{proof}

\subsection{Diagonalization of the ideal idempotent $T_0$}
\label{app:T0_diag}

\begin{proposition}[Diagonalization of $T_0$]
\label{prop:T0_diag}
Let $T_0$ be the noiseless reduced QED-cycle matrix in block form~\eqref{eq:T0_block_form},
ordered as $(L)\oplus(X)\oplus(R)$ with $L=2^k$ and $X=2^n-2^k$.  Then $T_0$ admits a diagonalization
$T_0 = VDV^{-1}$ with
\begin{equation}
V = \begin{bmatrix} I_L & 0 & 0 \\ 0 & 0 & I_X \\ 0 & 1 & -\mathbf{1}_X^T \end{bmatrix}, \quad
D = \begin{bmatrix} I_L & 0 & 0 \\ 0 & 1 & 0 \\ 0 & 0 & 0_X \end{bmatrix}, \quad
V^{-1} = \begin{bmatrix} I_L & 0 & 0 \\ 0 & \mathbf{1}_X^T & 1 \\ 0 & I_X & 0 \end{bmatrix}.
\end{equation}
Consequently,
\[
\sigma(T_0) = \{1,\ldots,1\}_{L+1} \cup \{0,\ldots,0\}_X.
\]
\end{proposition}

\begin{proof}
Direct verification shows $VV^{-1}=I$ and $T_0V=VD$.  The spectrum follows from the diagonal entries
of $D$.
\end{proof}

\begin{lemma}[Uniform condition number]
\label{lem:kappaV}
For the above diagonalization, $\|V\|_1 = \|V^{-1}\|_1 = 2$, hence $\kappa_1(V)=\|V\|_1\|V^{-1}\|_1=4$
independent of $(n,k)$.
\end{lemma}

\begin{proof}
Each logical column of $V$ has absolute column sum $1$, the rejection column has sum $1$, and each
leakage-eigenvector column has exactly two nonzero entries $\{+1,-1\}$ and hence absolute column sum
$2$, so $\|V\|_1=2$.

Similarly, each leakage column of $V^{-1}$ has two $+1$ entries (one in the $\mathbf{1}_X^T$ row and
one in the corresponding leakage row), giving absolute column sum $2$, while the remaining columns
have sum $1$, so $\|V^{-1}\|_1=2$.
\end{proof}

\subsection{Perturbative eigenvalue localization (Bauer--Fike)}
\label{app:bauer_fike}

\begin{proposition}[Bauer--Fike disks]
\label{prop:bauer_fike_disks}
Let $\Delta_{\varepsilon,\eta}^{(M,m)} := M_{\varepsilon,\eta}^{(m)}-T_0$ (Eq.~\eqref{eq:perturbation_notation}).
Every eigenvalue of $M_{\varepsilon,\eta}^{(m)}$ lies within distance $4\|\Delta_{\varepsilon,\eta}^{(M,m)}\|_1$
of $0$ or $1$:
\begin{equation}
\sigma(M_{\varepsilon,\eta}^{(m)})
\subseteq
\mathbb{D}\!\left(1, 4\|\Delta_{\varepsilon,\eta}^{(M,m)}\|_1\right)
\;\cup\;
\mathbb{D}\!\left(0, 4\|\Delta_{\varepsilon,\eta}^{(M,m)}\|_1\right),
\label{eq:BF_disks}
\end{equation}
where $\mathbb{D}(c,r):=\{z\in\mathbb{C}:|z-c|\le r\}$.
\end{proposition}

\begin{proof}
By Proposition~\ref{prop:T0_diag}, $T_0$ is diagonalizable: $T_0=VDV^{-1}$ with
$\sigma(T_0)\subseteq\{0,1\}$, and by Lemma~\ref{lem:kappaV}, $\kappa_1(V)=4$.

Apply the Bauer--Fike theorem~\cite{bauer1960norms} (in the induced $1$-norm): for any eigenvalue $\lambda$ of
$M_{\varepsilon,\eta}^{(m)}=T_0+\Delta_{\varepsilon,\eta}^{(M,m)}$, there exists $\mu\in\sigma(T_0)$ such that
\[
|\lambda-\mu|\le \kappa_1(V)\,\|\Delta_{\varepsilon,\eta}^{(M,m)}\|_1
=4\|\Delta_{\varepsilon,\eta}^{(M,m)}\|_1.
\]
Since $\mu\in\{0,1\}$, this implies the stated disk containment.
\end{proof}

\begin{remark}[Two spectral clusters at small noise]
If $4\|\Delta_{\varepsilon,\eta}^{(M,m)}\|_1<\tfrac12$, the disks around $0$ and $1$ are disjoint, and the
spectrum splits into an $(L+1)$-eigenvalue cluster near $1$ (slow accepted-shot decay modes) and an
$X$-eigenvalue cluster near $0$ (fast transients associated with leakage rejection).
\end{remark}

\subsection{Norm scaling for small noise}
\label{app:norm_scaling}

\begin{lemma}[Small-noise norm scaling (including $\Omega$ lower bounds)]
\label{lem:norm_scaling}
Fix $m\ge 1$.

\smallskip
\noindent\textbf{Upper bounds.}
Under finite fault-location Pauli-stochastic noise,
\[
\|\Delta_\varepsilon^{(T)}\|_1 = O(\varepsilon),
\qquad
\|\Delta_\eta^{(G,m)}\|_1 = O(m\eta),
\qquad
\|\Delta_{\varepsilon,\eta}^{(M,m)}\|_1 = O(\varepsilon + m\eta)
\]
for each fixed $m$.

\smallskip
\noindent\textbf{Lower bounds (nontrivial first-order noise).}
Suppose the QED-check noise is nontrivial at first order in the following concrete sense:
there exists a logical boundary label $q\in L$ and a \emph{single-fault} check pattern
$\sigma_\star$ (i.e., a weight-$1$ Pauli fault in the syndrome-extraction circuit, with no other
faults), occurring with probability $c_\star\varepsilon+O(\varepsilon^2)$ for some $c_\star>0$,
such that, starting from $(q,0^a)$, the faulted circuit is \emph{accepted} but the reduced
cycle does \emph{not} return to $q$; equivalently, the reduced single-fault kernel sends
$q\mapsto q_\star$ with $q_\star\in (L\cup X)$ and $q_\star\neq q$ (an undetected single fault).

Then $\|\Delta_\varepsilon^{(T)}\|_1=\Omega(\varepsilon)$, and hence
$\|\Delta_\varepsilon^{(T)}\|_1=\Theta(\varepsilon)$.

Likewise, suppose the gate-layer noise is nontrivial at first order in the sense that there exists
a data boundary label $q$ and a \emph{single-fault} gate pattern $\tau_\star$ in one logical-gate
layer, occurring with probability $c_\star\eta+O(\eta^2)$ for some $c_\star>0$,
such that $\tau_\star$ sends $q\mapsto q_\star$ with $q_\star\neq q$ (i.e., the induced Pauli
is not stabilizer-trivial on that label). Then for each fixed $m$,
$\|\Delta_\eta^{(G,m)}\|_1=\Omega(m\eta)$, and hence $\|\Delta_\eta^{(G,m)}\|_1=\Theta(m\eta)$.
\end{lemma}

\begin{proof}
\textbf{Upper bounds.}
Each entry of $T_\varepsilon$ is a finite weighted sum of fault-pattern probabilities, hence is
analytic in $\varepsilon$ at $0$; similarly for $G_\eta$ in $\eta$. Therefore
$T_\varepsilon=T_0+\varepsilon T_1+O(\varepsilon^2)$ and $G_\eta=I+\eta G_1+O(\eta^2)$ entrywise, giving
$\|T_\varepsilon-T_0\|_1=O(\varepsilon)$ and, for fixed $m$,
\[
(G_\eta)^m=(I+\eta G_1+O(\eta^2))^m = I + m\eta G_1 + O(m^2\eta^2),
\]
so $\|(G_\eta)^m-I\|_1=O(m\eta)$. Finally,
$M_{\varepsilon,\eta}^{(m)}=T_\varepsilon(G_\eta)^m$ and $\|T_\varepsilon\|_1=1$ (column-stochastic), so
\[
\|M_{\varepsilon,\eta}^{(m)}-T_0\|_1
\le \|T_\varepsilon-T_0\|_1 + \|T_\varepsilon\|_1\,\|(G_\eta)^m-I\|_1
= O(\varepsilon+m\eta).
\]

\medskip
\textbf{Lower bound for the check: $\|\Delta_\varepsilon^{(T)}\|_1=\Omega(\varepsilon)$.}
Let $e_q$ be the basis vector for the logical label $q\in L$. Since the ideal reduced QED cycle
preserves logical states, $T_0 e_q = e_q$.

By the stated nontriviality assumption, there exists a single-fault check pattern $\sigma_\star$
that is accepted and maps $q\mapsto q_\star\neq q$. Because the circuit-level faults are Pauli and
the underlying ideal stabilizer circuit is Clifford, each \emph{fixed} fault pattern induces a
\emph{deterministic} reduced transition $q\mapsto q_\star$ (or $q\mapsto R$ if rejected). Let
$p_\star(\varepsilon)$ be the probability that \emph{exactly} $\sigma_\star$ occurs (and no other
faults). Under Pauli-stochastic noise at finitely many locations,
$p_\star(\varepsilon)=c_\star \varepsilon + O(\varepsilon^2)$ for some constant $c_\star>0$.

Since $\sigma_\star$ never returns to $q$, it contributes \emph{zero} mass to the $q$-entry of
$T_\varepsilon e_q$, hence
\[
(T_\varepsilon e_q)_q \le 1 - p_\star(\varepsilon).
\]
Therefore,
\[
\|(T_\varepsilon-T_0)e_q\|_1
=\|T_\varepsilon e_q - e_q\|_1
= \bigl| (T_\varepsilon e_q)_q - 1 \bigr| + \sum_{i\neq q} (T_\varepsilon e_q)_i
= 2\bigl(1-(T_\varepsilon e_q)_q\bigr)
\ge 2p_\star(\varepsilon)
= 2c_\star\varepsilon + O(\varepsilon^2).
\]
Taking the induced matrix norm (max column $\ell_1$ sum) gives
$\|T_\varepsilon-T_0\|_1 \ge \|(T_\varepsilon-T_0)e_q\|_1 = \Omega(\varepsilon)$, as claimed.

\medskip
\textbf{Lower bound for gates: $\|\Delta_\eta^{(G,m)}\|_1=\Omega(m\eta)$.}
Let $e_q$ be a data-boundary basis vector. By the nontriviality assumption for gates, there exists
a single-fault gate pattern $\tau_\star$ in one gate layer that maps $q\mapsto q_\star\neq q$. Let
$p_\star(\eta)=c_\star \eta + O(\eta^2)$ be the probability that exactly $\tau_\star$ occurs in a
given gate layer (and no other faults in that layer), with $c_\star>0$.

Across $m$ independent gate layers, the event “$\tau_\star$ occurs in exactly one of the $m$ layers
and no other faults occur anywhere” has probability
$m\,p_\star(\eta) + O(\eta^2)= m c_\star\eta + O(\eta^2)$ for fixed $m$.
On that event, the final label is $q_\star\neq q$, so it contributes zero mass to the $q$-entry of
$(G_\eta)^m e_q$, implying
\[
\bigl((G_\eta)^m e_q\bigr)_q \le 1 - m c_\star\eta + O(\eta^2).
\]
Hence,
\[
\|((G_\eta)^m-I)e_q\|_1
= 2\bigl(1-\bigl((G_\eta)^m e_q\bigr)_q\bigr)
\ge 2 m c_\star\eta + O(\eta^2)
= \Omega(m\eta).
\]
Therefore $\|(G_\eta)^m-I\|_1 \ge \|((G_\eta)^m-I)e_q\|_1 = \Omega(m\eta)$.

Combining each lower bound with the corresponding upper bound yields the stated $\Theta$ scalings.
\end{proof}

\subsection{Probability flux through the ideal check}
\label{app:prob_flux_ideal}

Writing $\Delta_\eta^{(G)}:=G_\eta-I$ in $(L,X,R)$ block form and multiplying by $T_0$ yields:
\begin{equation}
T_0\Delta_\eta^{(G)} = \begin{bmatrix} \Delta^{(G)}_{LL} & \Delta^{(G)}_{LX} & 0 \\ 0 & 0 & 0 \\ \mathbf{1}_X^T\Delta^{(G)}_{XL} & \mathbf{1}_X^T\Delta^{(G)}_{XX} & 0 \end{bmatrix}.
\label{eq:T0_DeltaG_block}
\end{equation}
Gate noise contributes to rejection after an ideal check precisely through leakage creation: population pushed into leakage by $G_\eta$ is converted to rejection by $T_0$. The same block structure applies to $\Delta_\eta^{(G,m)}=(G_\eta^{(m)}-I)$, with the accepted-to-leakage block still vanishing.

\subsection{Probability flux identities}
\label{app:flux_identities}

Since $G_\eta$ is column-stochastic, every column of $\Delta_\eta^{(G)} = G_\eta - I$ sums to $0$. Separating logical vs.\ leakage columns gives the flux identities:
\begin{align}
\mathbf{1}_L^T \Delta^{(G)}_{LL} &= -\mathbf{1}_X^T \Delta^{(G)}_{XL}, \label{eq:prob_flux_1} \\
\mathbf{1}_L^T \Delta^{(G)}_{LX} &= -\mathbf{1}_X^T \Delta^{(G)}_{XX}. \label{eq:prob_flux_2}
\end{align}

Equation~\eqref{eq:prob_flux_1} states that since gate noise alone cannot enter $R$, any net population leaving the logical sector must enter the leakage sector. Under an ideal QED check, that leakage becomes rejection probability, which is exactly the mechanism encoded by the rejection-row block in~\eqref{eq:T0_DeltaG_block}.

\section{Motivation for the shot-efficiency (QED-efficiency) metric}
\label{app:se_efficiency_motivation}

This appendix \emph{restates} (with notation adapted to the present manuscript) the motivation for
the \emph{shot-efficiency} / \emph{QED-efficiency} metric introduced in
Kumar~\emph{et al.}~\cite{kumar2026co}.  The QED--PEC cost comparison and resulting efficiency condition originate in
\cite{kumar2026co}; we include this restatement here solely so that the present transition-matrix
manuscript is self-contained.

The motivation in~\cite{kumar2026co} is to quantify when using QED as a front-end can reduce the
\emph{total} number of circuit executions (shots) required by probabilistic error cancellation (PEC).
The key tradeoff is:

\begin{itemize}[leftmargin=*]
\item QED \emph{reduces} the effective error rate seen by downstream mitigation, but
\item QED \emph{rejects} a fraction of runs, increasing the number of raw shots needed to obtain a
fixed number of \emph{accepted} samples.
\end{itemize}

To express this tradeoff, we consider the $m$-units as in Section~\ref{sec:qed_efficiency}, each consisting of $m$ logical gates with a single noisy QED check at the end. For this cost model, let $\varepsilon_p$ be the relevant \emph{physical} error rate for the unit when run without QED
(post-selection absent), let $\varepsilon_\ell$ be the \emph{post-selected} logical error rate for the
same unit when run with QED (conditioning on acceptance), and let $p_{\mathrm{rej}}$ be the per-unit
rejection probability under QED.

\paragraph{PEC shot cost model.}
For this shot-cost comparison, successive accepted units are modeled by a
fixed logical Pauli channel with input-independent acceptance probability
$1-p_{\mathrm{rej}}$. The parameters $\varepsilon_p$ and $\varepsilon_\ell$
are the effective error parameters entering the per-unit quasi-probability
norms below.
A standard PEC scaling model is that to estimate an
observable to fixed additive precision, the number of shots scales as~\cite{temme2017error,van2023probabilistic}
\begin{equation}
C_{\mathrm{PEC}} = D\,\gamma^2,
\end{equation}
where $D$ hides precision- and constant-factor details and
$\gamma=\prod_{g=1}^M \gamma_g$ is the product of per-unit quasi-probability norms.
For small noise strength, one typically has the first-order approximation
\begin{equation}
\gamma_g \approx 1 + 2\varepsilon,
\end{equation}
so for a homogeneous unit repeated $M$ times,
\begin{equation}
C_{\mathrm{PEC}}(\varepsilon)
\;\approx\;
D(1+2\varepsilon)^{2M}
\;=\;
D\,\exp\!\bigl(4\varepsilon M + O(\varepsilon^2M)\bigr).
\label{eq:app_pec_cost_exp}
\end{equation}
Equation~\eqref{eq:app_pec_cost_exp} is the only PEC-specific approximation needed for the efficiency
condition below; see~\cite{kumar2026co} for the surrounding PEC context and assumptions.

\paragraph{Comparing PEC on physical qubits vs.\ PEC on post-selected QED.}
\emph{(i) PEC on the physical system.}
If one applies PEC directly to the underlying physical implementation (no QED), then
\begin{equation}
C_{\mathrm{PEC}}^{\mathrm{phys}}
\;\approx\;
D\,\exp\!\bigl(4\varepsilon_p M\bigr).
\label{eq:app_cost_phys}
\end{equation}

\emph{(ii) QED + PEC on the post-selected logical channel.}
If one first runs QED and post-selects, then only a fraction $(1-p_{\mathrm{rej}})^M$ of shots survive
over $M$ units. Thus, obtaining a fixed number of \emph{accepted} samples requires an overhead factor
\begin{equation}
\frac{1}{(1-p_{\mathrm{rej}})^M}
\;=\;
\exp\!\bigl(p_{\mathrm{rej}}M + O(p_{\mathrm{rej}}^2M)\bigr).
\label{eq:app_postselect_overhead}
\end{equation}
On those accepted samples, PEC now acts on the \emph{post-selected} (reduced-noise) channel with
effective error rate $\varepsilon_\ell$, giving a PEC cost
$D\exp(4\varepsilon_\ell M)$ by the same approximation as~\eqref{eq:app_pec_cost_exp}.
Multiplying by the post-selection overhead~\eqref{eq:app_postselect_overhead} yields
\begin{equation}
C_{\mathrm{PEC+QED}}
\;\approx\;
\exp\!\bigl(p_{\mathrm{rej}}M\bigr)\cdot D\,\exp\!\bigl(4\varepsilon_\ell M\bigr)
=
D\,\exp\!\bigl((p_{\mathrm{rej}}+4\varepsilon_\ell)M\bigr).
\label{eq:app_cost_qedpec}
\end{equation}

\paragraph{Advantage ratio and the per-unit efficiency condition.}
Define the (per-region) shot \emph{advantage ratio}
\begin{equation}
R
:=
\frac{C_{\mathrm{PEC}}^{\mathrm{phys}}}{C_{\mathrm{PEC+QED}}}.
\end{equation}
Using~\eqref{eq:app_cost_phys}--\eqref{eq:app_cost_qedpec} gives
\begin{equation}
R
\;\approx\;
\exp\!\Bigl(\bigl(4\varepsilon_p - 4\varepsilon_\ell - p_{\mathrm{rej}}\bigr)M\Bigr).
\label{eq:app_adv_ratio}
\end{equation}
Thus, in this leading-order model, QED+PEC uses fewer shots than physical-only PEC when $R>1$, i.e.,
when
\begin{equation}
4(\varepsilon_p-\varepsilon_\ell) > p_{\mathrm{rej}}
\qquad\Longleftrightarrow\qquad
\frac{\varepsilon_p-\varepsilon_\ell}{p_{\mathrm{rej}}} > \frac{1}{4}.
\label{eq:app_efficiency_condition}
\end{equation}

\paragraph{Definition of shot efficiency / QED efficiency.}
Motivated by Eq.~\eqref{eq:app_efficiency_condition}, Kumar~\emph{et al.}~\cite{kumar2026co} define
the \emph{shot-efficiency} / \emph{QED-efficiency} of a unit as the \emph{error reduction achieved per
unit rejection probability}:
\begin{equation}
\mathrm{SE}
:=
\frac{\varepsilon_p-\varepsilon_\ell}{p_{\mathrm{rej}}}.
\label{eq:app_se_def}
\end{equation}
In the above PEC scaling approximation, the condition $\mathrm{SE}>\tfrac{1}{4}$ is exactly the
per-unit criterion that makes the exponent in~\eqref{eq:app_adv_ratio} positive, and hence yields an
asymptotic total-shot advantage over $M$ repeated units.

\paragraph{Connection to the present manuscript.}
Section~\ref{sec:qed_efficiency} studies the same efficiency functional
using the operational $m$-unit rates of
Definitions~\ref{def:metrics_m_unit} and~\ref{def:physical_error_rate}:
\[
\mathrm{SE}(m)=\frac{\varepsilon_p^{(m)}-\varepsilon_\ell^{(m)}}{p_{\mathrm{rej}}^{(m)}}.
\]
Under the shot-cost assumptions above, the PEC advantage criterion applies
when these rates agree, to the retained order, with the effective-channel
parameters above. The transition-matrix expansion of $\mathrm{SE}(m)$
does not require this additional identification.
The remainder of Appendix~\ref{app:se_efficiency_details} in this manuscript then develops the
transition-matrix perturbative analysis of $\mathrm{SE}(m)$ as a function of $m$ under circuit-level
Pauli-stochastic noise.

\section{Deferred derivations for Subsection~\ref{sec:se_firstord}}
\label{app:se_efficiency_details}

\subsection{Existence of first-order expansions and proof of Lemma~\ref{lem:se_A1_affine}}

\paragraph{First-order expansions under Pauli-stochastic noise.}
In a circuit-level stochastic Pauli model with per-location error probability $\varepsilon$, each
fault pattern $\sigma$ has probability $\Pr_\varepsilon(\sigma)$ that is a polynomial in $\varepsilon$
(e.g., $\Pr_\varepsilon(\sigma)=\varepsilon^{|\sigma|}(1-\varepsilon)^{F-|\sigma|}$ in the i.i.d.\ model),
and the transition matrices are finite sums over patterns. Hence each matrix entry of $T_\varepsilon$
and $G_\varepsilon$ is analytic at $\varepsilon=0$, so the Taylor expansions
$T_\varepsilon=T_0+\varepsilon T_1+O(\varepsilon^2)$ and $G_\varepsilon=I+\varepsilon G_1+O(\varepsilon^2)$ exist, with
$T_1=\left.\frac{d}{d\varepsilon}T_\varepsilon\right|_{\varepsilon=0}$ and
$G_1=\left.\frac{d}{d\varepsilon}G_\varepsilon\right|_{\varepsilon=0}$.

\paragraph{Remainder for an $m$-gate interval.}
Because an $m$-gate interval contains $O(m)$ independent fault locations, the order-$\varepsilon^2$
Taylor coefficients arise from choosing two fault locations and hence scale as $O(m^2)$. In
particular,
\[
(G_\varepsilon)^m = I + m\varepsilon G_1 + O(m^2\varepsilon^2),
\]
entrywise (and hence in any induced matrix norm). Consequently,
\[
M^{(m)}_\varepsilon=T_\varepsilon(G_\varepsilon)^m
= T_0 + \varepsilon(T_1+mT_0G_1)+O(m^2\varepsilon^2),
\]
and taking accepted blocks yields
$A^{(m)}_\varepsilon=A_0+\varepsilon A_1^{(m)}+O(m^2\varepsilon^2)$.

\begin{proof}[Full proof of Lemma~\ref{lem:se_A1_affine}]
Write $M^{(m)}_\varepsilon=T_\varepsilon(G_\varepsilon)^m$. Differentiate at $\varepsilon=0$:
\[
\left.\frac{d}{d\varepsilon}M^{(m)}_\varepsilon\right|_{\varepsilon=0}
=
\left.\frac{d}{d\varepsilon}T_\varepsilon\right|_{\varepsilon=0}(G_0)^m
+
T_0\left.\frac{d}{d\varepsilon}(G_\varepsilon^m)\right|_{\varepsilon=0}.
\]
Since $G_0=I$, it remains to compute $\left.\frac{d}{d\varepsilon}(G_\varepsilon^m)\right|_{\varepsilon=0}$.
Using the product rule for $m$ factors,
\[
\frac{d}{d\varepsilon}(G_\varepsilon^m)
=
\sum_{j=0}^{m-1} G_\varepsilon^{\,j}\left(\frac{dG_\varepsilon}{d\varepsilon}\right)G_\varepsilon^{\,m-1-j}.
\]
Evaluating at $\varepsilon=0$ gives
$\left.\frac{d}{d\varepsilon}(G_\varepsilon^m)\right|_{\varepsilon=0}=\sum_{j=0}^{m-1}I^jG_1I^{m-1-j}=mG_1$.
Therefore
$\left.\frac{d}{d\varepsilon}M^{(m)}_\varepsilon\right|_{\varepsilon=0}=T_1+mT_0G_1$.
Taking accepted blocks is linear, so
$A_1^{(m)}=(T_1)_{\mathrm{acc}}+m(T_0G_1)_{\mathrm{acc}}$.
\end{proof}

\subsection{Two-cycle ratio expansions (Lemma~\ref{lem:se_two_cycle})}

We use the first-order ratio identity
\begin{equation}
\frac{1+\varepsilon u + O(\varepsilon^2)}{1+\varepsilon v + O(\varepsilon^2)}
=
1+\varepsilon(u-v)+O(\varepsilon^2),
\label{eq:app_ratio_id}
\end{equation}
which follows from $(1+\varepsilon v)^{-1}=1-\varepsilon v+O(\varepsilon^2)$.

\begin{proof}[Full proof of Lemma~\ref{lem:se_two_cycle}]
Let $A_\varepsilon=A_0+\varepsilon A_1+O(\varepsilon^2)$ with $A_0=\mathrm{diag}(I_L,0_X)$, so $A_0^2=A_0$ and
$A_0e_\psi=e_\psi$. (When applying this lemma to $m$-units, Lemma~\ref{lem:se_A1_affine} gives the
sharper expansion $A^{(m)}_\varepsilon=A_0+\varepsilon A_1^{(m)}+O(m^2\varepsilon^2)$, and the same
algebra below yields the conclusions with $O(m^2\varepsilon^2)$ remainders.)

\paragraph{Step 1: expand $A_\varepsilon e_\psi$ and $A_\varepsilon^2 e_\psi$.}
We have
\[
A_\varepsilon e_\psi = e_\psi + \varepsilon A_1e_\psi + O(\varepsilon^2),
\]
and
\[
A_\varepsilon^2 = (A_0+\varepsilon A_1)^2 + O(\varepsilon^2)
= A_0 + \varepsilon(A_0A_1 + A_1A_0) + O(\varepsilon^2),
\]
so
\[
A_\varepsilon^2 e_\psi
= e_\psi + \varepsilon(A_0A_1e_\psi + A_1e_\psi) + O(\varepsilon^2).
\]

\paragraph{Step 2: expand $p_{\mathrm{acc}}(1)$ and $p_{\mathrm{acc}}(2)$.}
Since $A_\varepsilon$ is substochastic and $e_\psi\ge 0$,
$p_{\mathrm{acc}}(t)=\|A_\varepsilon^t e_\psi\|_1=\mathbf{1}_{L+X}^T A_\varepsilon^t e_\psi$.
Thus
\[
p_{\mathrm{acc}}(1)=1+\varepsilon\,\mathbf{1}_{L+X}^T A_1e_\psi + O(\varepsilon^2),
\]
and using $\mathbf{1}_{L+X}^TA_0=\mathbf{1}_L^T$,
\[
p_{\mathrm{acc}}(2)=1+\varepsilon\Big(\mathbf{1}_{L}^TA_1e_\psi+\mathbf{1}_{L+X}^TA_1e_\psi\Big)+O(\varepsilon^2).
\]
Apply \eqref{eq:app_ratio_id} to $p_{\mathrm{rej}}=1-\frac{p_{\mathrm{acc}}(2)}{p_{\mathrm{acc}}(1)}$,
to obtain
\[
p_{\mathrm{rej}}=-\varepsilon\,\mathbf{1}_L^TA_1e_\psi+O(\varepsilon^2)
=-\varepsilon\,\mathbf{1}_L^T(A_1)_{LL}e_\psi+O(\varepsilon^2).
\]

\paragraph{Step 3: expand $\mathrm{ASC}(1)$ and $\mathrm{ASC}(2)$.}
For $t=1$,
\[
e_\psi^TA_\varepsilon e_\psi = 1+\varepsilon\,e_\psi^TA_1e_\psi+O(\varepsilon^2),
\]
so
\[
\mathrm{ASC}(1)
=
\frac{1+\varepsilon\,e_\psi^TA_1e_\psi+O(\varepsilon^2)}{1+\varepsilon\,\mathbf{1}_{L+X}^TA_1e_\psi+O(\varepsilon^2)}
=
1+\varepsilon\big(e_\psi^TA_1e_\psi-\mathbf{1}_{L+X}^TA_1e_\psi\big)+O(\varepsilon^2).
\]
For $t=2$, use $e_\psi^TA_0=e_\psi^T$ and $A_0e_\psi=e_\psi$ to get
\[
e_\psi^TA_\varepsilon^2e_\psi
=
1+\varepsilon\,e_\psi^T(A_0A_1+A_1A_0)e_\psi+O(\varepsilon^2)
=
1+2\varepsilon\,e_\psi^TA_1e_\psi+O(\varepsilon^2),
\]
and hence
\[
\mathrm{ASC}(2)
=
\frac{1+2\varepsilon\,e_\psi^TA_1e_\psi+O(\varepsilon^2)}{1+\varepsilon(\mathbf{1}_{L}^TA_1e_\psi+\mathbf{1}_{L+X}^TA_1e_\psi)+O(\varepsilon^2)}
=
1+\varepsilon\big(2e_\psi^TA_1e_\psi-\mathbf{1}_{L}^TA_1e_\psi-\mathbf{1}_{L+X}^TA_1e_\psi\big)+O(\varepsilon^2).
\]
Apply \eqref{eq:app_ratio_id} to $\frac{\mathrm{ASC}(2)}{\mathrm{ASC}(1)}$ and simplify to obtain
\[
\varepsilon_\ell
=
1-\frac{\mathrm{ASC}(2)}{\mathrm{ASC}(1)}
=
\varepsilon\big(\mathbf{1}_L^TA_1e_\psi-e_\psi^TA_1e_\psi\big)+O(\varepsilon^2),
\]
which equals $\varepsilon\sum_{\substack{i\in L\\ i\neq \psi}}(A_1)_{i,\psi}+O(\varepsilon^2)$ on the $LL$ block.
\end{proof}

\subsection{Gate faults do not contribute to $\varepsilon_\ell^{(m)}$ at first order (details for Proposition~\ref{prop:se_first_order_coeffs})}

Let $P$ be the effective Pauli induced by a \emph{single} gate fault in the interval. By
Assumption~\ref{ass:eff_gate}, $\mathrm{wt}(P)\le c_{\mathrm{prop}}<d$.
If $P$ has nontrivial syndrome it is rejected by the ideal check $T_0$.
If $P$ has trivial syndrome, then $P\in N(\mathcal{S})$. But any nontrivial logical Pauli has
minimum weight $d$, so $\mathrm{wt}(P)<d$ forces $P\in\mathcal{S}$, hence it acts trivially on the
logical label. Therefore at $O(\varepsilon)$ the gate contribution $(T_0G_1)_{\mathrm{acc}}$ has no
off-diagonal logical transitions, implying the $m$-dependent term vanishes in
\eqref{eq:se_ell_from_A1}. This completes the $m$-independence claim.

\subsection{Remainder bookkeeping for Theorem~\ref{thm:se_first_order}}
\label{app:se_first_order_remainder}
Using $\varepsilon_p^{(m)}=m\varepsilon+O(m^2\varepsilon^2)$,
$p_{\mathrm{rej}}^{(m)}=\varepsilon(\alpha+\beta m)+O(m^2\varepsilon^2)$, and
$\varepsilon_\ell^{(m)}=\gamma\varepsilon+O(m^2\varepsilon^2)$, we have
\[
\mathrm{SE}(m)
=
\frac{\varepsilon(m-\gamma)+O(m^2\varepsilon^2)}{\varepsilon(\alpha+\beta m)+O(m^2\varepsilon^2)}
=
\frac{m-\gamma+O(m^2\varepsilon)}{\alpha+\beta m+O(m^2\varepsilon)}.
\]
Expanding the ratio to first order gives
\[
\mathrm{SE}(m)
=
\frac{m-\gamma}{\alpha+\beta m}
+\;O\!\left(
\varepsilon\,\frac{m^2}{\alpha+\beta m}
+\varepsilon\,\frac{m^3}{(\alpha+\beta m)^2}
\right)
=
\frac{m-\gamma}{\alpha+\beta m}
+\;O\!\left(
\varepsilon\,\frac{m^2}{\alpha+\beta m}\left(1+\frac{m}{\alpha+\beta m}\right)
\right).
\]

which matches Theorem~\ref{thm:se_first_order}.

\section{Supplementary details for Section~\ref{sec:symmetry_model_reduction}: Symmetry-Based Model Reduction}
\label{app:symmetry_reduction_details}

This appendix collects (i) brief classical lumping background (textbook material),
(ii) the quantum-to-classical symmetry lemmas used in Proposition~\ref{prop:ancilla_clean_symmetry},
and (iii) extended orbit identification for the $[[4,2,2]]$ example.

\subsection{Main symmetry-to-orbit-quotient statement (formal)}
\label{app:main_symmetry_orbit_statement}

\begin{proposition}[Protocol symmetries induce exact orbit quotients under homogeneous ancilla-clean noise]
\label{prop:ancilla_clean_symmetry}
Let $G$ be a group of implemented wire-permutation symmetries in the
fault-path sense specified in Section~\ref{sec:symmetry_model_reduction},
and let $U_g$ denote the corresponding permutation unitary for $g\in G$.
Assume furthermore that:

\begin{enumerate}
\item $U_g$ is monomial in the joint basis $\mathcal{B}$, hence induces a permutation $\pi_g'$ of joint
labels $(q,y)$.

\item $U_g$ preserves the QED acceptance set
\[
\mathcal{A}:=\{(q,0^a): q\in\mathcal{B}_{\mathcal{C}}\},
\]
so that $\pi_g'$ restricts to a permutation of accepted boundary labels and induces a permutation
$\pi_g$ of $\overline{\mathcal{Q}}$ (with $\pi_g(R)=R$).
\end{enumerate}

Under Assumption~\ref{assump:ancilla_clean_checks}, $T_\varepsilon=\mathcal{R}(T'_\varepsilon)$ and $G_\eta$
are $\pi_g$-invariant for all $g\in G$, and hence every clock-cycle matrix
\[
M^{(m)}_{\varepsilon,\eta} \;=\; T_\varepsilon (G_\eta)^m
\]
is $\pi_g$-invariant. Consequently, the orbit partition of $\overline{\mathcal{Q}}$ under $G$ is
strongly lumpable and yields an \emph{exact} quotient chain whose dimension equals the number of
$G$-orbits.
\end{proposition}

\noindent A proof appears in Subsection~\ref{app:ancilla_clean_proofs} after we collect the needed
background lemmas.

\subsection{Strong lumpability and orbit quotients (Background)}
\label{app:lumping_basic_proofs}

\begin{definition}[Strong lumpability for column-(sub)stochastic kernels]
\label{def:strong_lumpability}
Let $\Omega$ be finite and let $M\ge 0$ be column-substochastic on $\Omega$.
Let $\mathcal{P}=\{C_1,\dots,C_m\}$ be a partition of $\Omega$ into nonempty blocks.
Define the lumping map $L\in\{0,1\}^{m\times|\Omega|}$ by
$L_{a,i}:=\mathbf{1}[i\in C_a]$.
We say $\mathcal{P}$ is \emph{strongly lumpable} for $M$ if there exists a nonnegative
column-substochastic $\overline{M}\in\mathbb{R}^{m\times m}$ such that
\[
LM=\overline{M}L.
\]
\end{definition}

\begin{theorem}[Orbit lumping (standard)]
\label{thm:orbit_lumping_app}
Let $M$ be nonnegative column-substochastic on $\Omega$, and let $G\le S(\Omega)$ be a symmetry group
of $M$ in the sense of Definition~\ref{def:matrix_symmetry}. Partition $\Omega$ into $G$-orbits
$\Omega=\bigsqcup_{a=1}^m \mathcal{O}_a$. Then the orbit partition is strongly lumpable for $M$.
Moreover, the quotient kernel can be taken as
\[
\overline{M}_{a,b}:=\sum_{i\in\mathcal{O}_a} M_{i,j}\qquad\text{for any } j\in\mathcal{O}_b.
\]
\end{theorem}

\begin{proof}[Proof sketch / reference]
We apply the lumpability criterion of Kemeny and Snell~\cite[Section~6.3, Theorem~6.3.2, p.~124]{kemeny1976finite}. Invariance implies that for any two states
$j,j'$ in the same orbit and any fixed orbit $\mathcal{O}_a$, the total mass sent into $\mathcal{O}_a$
is the same, so the orbit partition satisfies the strong-lumpability condition.
\end{proof}

\begin{lemma}[Which observables are preserved by a given lumping]
\label{lem:observable_preserved}
Let $\mathcal{P}=\{C_1,\dots,C_m\}$ be a partition of $\Omega$ with lumping map $L$.
A function $f:\Omega\to\mathbb{R}$ (viewed as a vector in $\mathbb{R}^{|\Omega|}$) factors through the
lumped state if and only if it is constant on each block; equivalently, iff there exists
$\overline{f}\in\mathbb{R}^m$ such that $f=L^T\overline{f}$. In that case $f^Tv=\overline{f}^{\,T}(Lv)$
depends only on coarse populations.
\end{lemma}

\begin{proof}
Immediate from the definition of $L$.
\end{proof}

\subsection{Covariance implies symmetry of the induced transition matrix}
\label{app:covariance_symmetry_proofs}

\begin{lemma}[Covariance $\Rightarrow$ symmetry of the induced transition matrix]
\label{lem:covariance_implies_TM_symmetry}
Let $\mathcal{B}=\{\ket{b_i}\}$ be an orthonormal basis with projectors $\Pi_i=\ket{b_i}\!\bra{b_i}$.
Suppose a unitary $U$ is monomial in $\mathcal{B}$, so $U\ket{b_i}=e^{i\theta_i}\ket{b_{\pi(i)}}$ for a
permutation $\pi$. If a CP map $\mathcal{E}$ is $U$-covariant,
$\mathcal{E}(U\rho U^\dagger)=U\mathcal{E}(\rho)U^\dagger$ for all $\rho$, then
\[
T_{\mathcal{B}}(\mathcal{E}) = P_\pi\,T_{\mathcal{B}}(\mathcal{E})\,P_\pi^{-1}.
\]
\end{lemma}

\begin{proof}
Monomiality gives $U\Pi_iU^\dagger=\Pi_{\pi(i)}$. Covariance gives
$\mathcal{E}(\Pi_{\pi(j)})=\mathcal{E}(U\Pi_jU^\dagger)=U\mathcal{E}(\Pi_j)U^\dagger$.
Then, using cyclicity of trace,
\begin{align*}
T_{\mathcal{B}}(\mathcal{E})_{\pi(i),\pi(j)}
&=\Tr\!\bigl(\Pi_{\pi(i)}\,\mathcal{E}(\Pi_{\pi(j)})\bigr)
=\Tr\!\bigl(U\Pi_iU^\dagger\;U\mathcal{E}(\Pi_j)U^\dagger\bigr) \\
&=\Tr\!\bigl(U\Pi_i\,\mathcal{E}(\Pi_j)\,U^\dagger\bigr)
=\Tr\!\bigl(\Pi_i\,\mathcal{E}(\Pi_j)\bigr)
=T_{\mathcal{B}}(\mathcal{E})_{i,j}.
\end{align*}
Thus $T_{\pi(i),\pi(j)}=T_{i,j}$ for all $i,j$, which is equivalent to the stated conjugation symmetry.
\end{proof}

\subsection{A fault-path relabeling sufficient condition for covariance}
\label{app:fault_path_covariance}

\begin{proposition}[Fault-path relabeling criterion for covariance]
\label{prop:fault_path_covariance}
Consider a circuit with fault locations $\mathcal{L}$ and Pauli fault patterns
$\Sigma=\{(P_\ell)_{\ell\in\mathcal{L}}\}$, with $P_\ell\in\{I,X,Y,Z\}$.
Suppose the implemented noisy map decomposes as
\[
\mathcal{E} = \sum_{\sigma\in\Sigma} \Pr(\sigma)\,\mathcal{E}_\sigma,
\]
where $\mathcal{E}_\sigma$ is obtained by inserting the Pauli faults specified by $\sigma$ into the
ideal circuit. Let $U$ be a unitary. If there exists a bijection $\phi:\Sigma\to\Sigma$ such that
\[
\mathcal{E}_{\phi(\sigma)}(U\rho U^\dagger)=U\,\mathcal{E}_\sigma(\rho)\,U^\dagger
\qquad\forall\sigma,\ \forall\rho,
\]
and $\Pr(\phi(\sigma))=\Pr(\sigma)$ for all $\sigma$, then $\mathcal{E}$ is $U$-covariant.
\end{proposition}

\begin{proof}
Compute, for any $\rho$,
\begin{align*}
\mathcal{E}(U\rho U^\dagger)
&=\sum_{\sigma\in\Sigma}\Pr(\sigma)\,\mathcal{E}_\sigma(U\rho U^\dagger)
=\sum_{\sigma\in\Sigma}\Pr(\sigma)\,U\mathcal{E}_{\phi^{-1}(\sigma)}(\rho)U^\dagger \\
&=\sum_{\sigma'\in\Sigma}\Pr(\phi(\sigma'))\,U\mathcal{E}_{\sigma'}(\rho)U^\dagger
=U\Bigl(\sum_{\sigma'\in\Sigma}\Pr(\sigma')\,\mathcal{E}_{\sigma'}(\rho)\Bigr)U^\dagger
=U\,\mathcal{E}(\rho)\,U^\dagger.
\end{align*}
\end{proof}

\subsection{Symmetry descent through QED reduction}
\label{app:qed_reduction_symmetry_proofs}

\begin{lemma}[QED reduction preserves acceptance-set symmetries]
\label{lem:reduction_preserves_symmetry}
Let $T'$ be a column-stochastic matrix on joint labels $(q,y)\in\mathcal{B}_{\mathcal{C}}\times\{0,1\}^a$
and let $T=\mathcal{R}(T')$ be its QED reduction (Definition~\ref{def:reduction_map}).
Let $\pi'$ be a permutation of joint labels that preserves
$\mathcal{A}=\{(q,0^a)\}$. Define an induced permutation $\pi$ on $\overline{\mathcal{Q}}$ by
$\pi(q)=q'$ whenever $\pi'(q,0^a)=(q',0^a)$ and $\pi(R)=R$.
If $T'$ is $\pi'$-invariant, then $T$ is $\pi$-invariant.
\end{lemma}

\begin{proof}
For $q,q'\in\mathcal{B}_{\mathcal{C}}$, QED reduction gives
$T[q',q]=T'[(q',0^a),(q,0^a)]$.
If $T'$ is $\pi'$-invariant and $\pi'$ preserves $\mathcal{A}$, then
\[
T'[(q',0^a),(q,0^a)] = T'[\pi'(q',0^a),\pi'(q,0^a)]
= T'[(\pi(q'),0^a),(\pi(q),0^a)],
\]
so $T[\pi(q'),\pi(q)]=T[q',q]$, i.e.\ the accepted block is $\pi$-invariant.
The rejection row is determined by accepted-column sums, hence is also invariant; and $R$ is absorbing.
\end{proof}

\subsection{Proof of Proposition~\ref{prop:ancilla_clean_symmetry}}
\label{app:ancilla_clean_proofs}

\begin{proof}[Proof of Proposition~\ref{prop:ancilla_clean_symmetry}]
By the implemented-symmetry hypothesis and
Assumption~\ref{assump:ancilla_clean_checks}, for each $g\in G$ the fault-path
relabeling induced by $U_g$ preserves probabilities and conjugates the
faulted maps, so Proposition~\ref{prop:fault_path_covariance} implies that both
$\mathcal{G}_\eta$ and $\mathcal{T}_\varepsilon$ are $U_g$-covariant. Hence their composition
$\mathcal{E}_{\varepsilon,\eta}=\mathcal{T}_\varepsilon\circ\mathcal{G}_\eta$ is $U_g$-covariant.

Because $U_g$ is monomial in the joint basis $\mathcal{B}$, Lemma~\ref{lem:covariance_implies_TM_symmetry}
implies $T_{\mathcal{B}}(\mathcal{E}_{\varepsilon,\eta})$ is invariant under the induced permutation
$\pi_g'$ of joint labels. Since $\pi_g'$ preserves the acceptance set $\mathcal{A}$, Lemma~\ref{lem:reduction_preserves_symmetry}
implies the reduced kernel $M_{\varepsilon,\eta}=\mathcal{R}(T_{\mathcal{B}}(\mathcal{E}_{\varepsilon,\eta}))$
is invariant under the induced permutation $\pi_g$ on $\overline{\mathcal{Q}}$ (fixing $R$).
Orbit lumping (Theorem~\ref{thm:orbit_lumping_app}) then yields strong lumpability of the orbit partition.
\end{proof}

\subsection{Example: the $[[4,2,2]]$ code orbit quotient under $S_4$}
\label{app:422_example}

We provide explicit representatives and a proof of the orbit decomposition summarized in
Example~\ref{ex:422_quotient}.

The $[[4,2,2]]$ code has stabilizer generators
\[
S_X := X^{\otimes 4}, \qquad S_Z := Z^{\otimes 4}.
\]
Fix logical codewords (one convenient choice) as
\begin{align*}
\ket{00}_L &= \tfrac{1}{\sqrt{2}}(\ket{0000}+\ket{1111}), &
\ket{01}_L &= \tfrac{1}{\sqrt{2}}(\ket{0011}+\ket{1100}), \\
\ket{10}_L &= \tfrac{1}{\sqrt{2}}(\ket{0101}+\ket{1010}), &
\ket{11}_L &= \tfrac{1}{\sqrt{2}}(\ket{0110}+\ket{1001}).
\end{align*}
There are $n-k=2$ syndrome bits, hence $4$ syndrome sectors. Choose representatives
$E_{00}=I$, $E_{10}=Z_1$, $E_{01}=X_1$, $E_{11}=Y_1$ and form the code-adapted basis
$\ket{s,\ell}:=E_s\ket{\ell}_L$.

\begin{proposition}[$S_4$-orbit lumping yields an exact $7\times 7$ quotient chain]
\label{prop:422_orbit_lumping}
Assume the implemented protocol/noise satisfy the symmetry conditions of
Proposition~\ref{prop:ancilla_clean_symmetry} for $G=S_4$ acting by data-wire permutations. Then the
reduced kernel $M_{\varepsilon,\eta}$ is $S_4$-invariant, hence the $S_4$-orbit partition of
$\overline{\mathcal{Q}}=\mathcal{B}_{\mathcal{C}}\cup\{R\}$ is strongly lumpable.

With the above basis, the $S_4$-orbits in the accepted space $\mathcal{B}_{\mathcal{C}}$ are exactly:
\begin{align*}
\mathcal{O}_1 &:= \{\ket{00}_L\}, \\
\mathcal{O}_2 &:= \{\ket{01}_L,\ket{10}_L,\ket{11}_L\}, \\
\mathcal{O}_3 &:= \{Z_1\ket{00}_L\}, \\
\mathcal{O}_4 &:= \{Z_1\ket{01}_L,\ Z_1\ket{10}_L,\ Z_1\ket{11}_L\}, \\
\mathcal{O}_5 &:= \{X_1\ket{\ell}_L : \ell\in\{00,01,10,11\}\}, \\
\mathcal{O}_6 &:= \{Y_1\ket{\ell}_L : \ell\in\{00,01,10,11\}\},
\end{align*}
together with $\mathcal{O}_7:=\{R\}$.
\end{proposition}

\begin{proof}
$S_4$-invariance and strong lumpability follow from Proposition~\ref{prop:ancilla_clean_symmetry} and
Theorem~\ref{thm:orbit_lumping_app}. It remains to identify the orbits.

$\ket{00}_L$ is invariant under all qubit permutations (it is GHZ$^+$), so it is a singleton orbit.
The three other logical codewords correspond to the three perfect matchings of $\{1,2,3,4\}$, and
$S_4$ acts transitively on perfect matchings, giving $\mathcal{O}_2$.

For $\ket{00}_L$, $Z_j\ket{00}_L$ is independent of $j$ up to phase, giving singleton orbit $\mathcal{O}_3$.
Applying $Z_1$ to $\{\ket{01}_L,\ket{10}_L,\ket{11}_L\}$ yields a set permuted transitively by $S_4$,
giving $\mathcal{O}_4$.

For the $X$-syndrome sector, transpositions in $S_4$ move $X_1$ across data wires, so the orbit of
$X_1\ket{00}_L$ contains all four states $\{X_1\ket{\ell}_L\}_\ell$, and cannot be larger because the
sector has size four. Thus $\mathcal{O}_5$ is the entire $X$-sector. The same argument applies to the
$Y$-sector, giving $\mathcal{O}_6$. Finally, $R$ is absorbing and fixed, so $\mathcal{O}_7=\{R\}$.
\end{proof}

\subsection{Reachability reductions under strict bias}
\label{app:reachability}

\begin{definition}[Reachable set and invariant restriction]
\label{def:reachable_set}
Let $M$ be a nonnegative matrix on a finite state set $\Omega$. For an initial population vector
$v^{(0)}\ge 0$, define the reachable set
\[
\Omega_{\mathrm{reach}}(v^{(0)}) := \bigl\{ i\in\Omega : (M^t v^{(0)})_i > 0 \text{ for some } t\ge 0 \bigr\}.
\]
A subset $S\subseteq\Omega$ is \emph{$M$-invariant} if $M_{i,j}=0$ for all $j\in S$ and $i\notin S$.
\end{definition}

\begin{lemma}[Exact restriction to an invariant reachable set]
\label{lem:reachable_restriction}
Let $M$ be nonnegative and let $S\subseteq\Omega$ be $M$-invariant and contain the support of
$v^{(0)}$. Let $M|_S$ denote the restriction of $M$ to rows/columns indexed by $S$. Then for all
$t\ge 0$, the restriction of $M^t v^{(0)}$ to $S$ equals $(M|_S)^t (v^{(0)}|_S)$.
\end{lemma}

\begin{proof}
Because $S$ is invariant, multiplying by $M$ cannot create support outside $S$ when the input is
supported on $S$. The identity follows by induction on $t$.
\end{proof}

\subsection{Approximate lumping under approximate symmetry}
\label{app:approx_lumping}

\begin{lemma}[A perturbative bound for approximate lumping]
\label{lem:approx_lumping_bound}
Let $M\in\mathbb{R}^{N\times N}$ be nonnegative and column-substochastic and let $L$ be the lumping
map for some partition $\mathcal{P}$. Fix any nonnegative column-substochastic
$\overline{M}\in\mathbb{R}^{m\times m}$ and define the \emph{lumping defect}
\[
\delta := \|LM-\overline{M}L\|_1.
\]
Then for every $t\ge 0$,
\[
\|L M^t - \overline{M}^{\,t} L\|_1 \;\le\; t\,\delta.
\]
Consequently, for any $v^{(0)}\ge 0$,
\[
\|L M^t v^{(0)} - \overline{M}^{\,t} L v^{(0)}\|_1 \le t\,\delta\,\|v^{(0)}\|_1.
\]
\end{lemma}

\begin{proof}
Let $E_t:=LM^t-\overline{M}^{\,t}L$. Then $E_0=0$ and
\[
E_{t+1}=(LM-\overline{M}L)M^t+\overline{M}E_t.
\]
Taking induced $1$-norms and using $\|M\|_1\le 1$ and $\|\overline{M}\|_1\le 1$ for column-substochastic
nonnegative matrices yields $\|E_{t+1}\|_1\le \delta+\|E_t\|_1$, hence $\|E_t\|_1\le t\delta$ by induction.
\end{proof}

\begin{remark}[A canonical ``best symmetric'' coarse model via group averaging]
\label{rem:symmetrize_then_lump}
Suppose a target symmetry group $G$ acts on $\Omega$ via permutation matrices $\{P_g\}_{g\in G}$, but
the true kernel $M$ is only approximately $G$-invariant. A canonical exactly $G$-invariant
approximation is the group-averaged (symmetrized) kernel
\[
M^{(G)} \;:=\; \frac{1}{|G|}\sum_{g\in G} P_g\,M\,P_g^{-1},
\]
which satisfies $P_h M^{(G)} P_h^{-1}=M^{(G)}$ for all $h\in G$.

Let $L$ be the orbit-lumping map for the $G$-orbit partition and let $\overline{M}$ be the exact
orbit quotient of $M^{(G)}$, so that $L M^{(G)}=\overline{M}L$. Then the lumping defect of $M$
relative to this coarse model obeys
\[
\delta \;=\;\|LM-\overline{M}L\|_1
\;=\;\|L(M-M^{(G)})\|_1
\;\le\;\|M-M^{(G)}\|_1,
\]
since $\|L\|_1=1$ for an orbit-sum map. Combining this estimate with Lemma~\ref{lem:approx_lumping_bound}
gives a concrete way to translate ``how badly symmetry is broken'' (measured by $\|M-M^{(G)}\|_1$)
into an explicit bound on coarse-grained prediction error over $t$ cycles.
\end{remark}

\section{Example: the $[[7,1,3]]$ Steane code orbit quotient under $\mathrm{GL}(3,2)$}
\label{app:steane_orbit_lumping}

This appendix gives a second concrete symmetry-based lumping example (beyond
Example~\ref{ex:422_quotient}) to demonstrate that the permutation/orbit reduction is not special to
the $[[4,2,2]]$ code.  We use the $[[7,1,3]]$ Steane code, which has a large data-wire permutation
automorphism group (order $168$).  The reduced QED state space has $2^7+1=129$ states, and orbit
lumping under this symmetry group yields an \emph{exact} $11$-state quotient chain.

\subsection{Steane code presentation and its data-wire permutation symmetries}

The Steane code is a self-dual CSS code derived from the classical $[7,4,3]$ Hamming code.  Fix the
$3\times 7$ parity-check matrix
\begin{equation}
H \;=\;
\begin{bmatrix}
1&1&1&1&0&0&0\\
1&1&0&0&1&1&0\\
1&0&1&0&1&0&1
\end{bmatrix},
\label{eq:steane_H}
\end{equation}
whose columns are precisely the seven nonzero vectors in $\mathbb{F}_2^3$.

Define the Steane stabilizer generators (three $X$-type and three $Z$-type) by
\begin{equation}
S_i^X \;:=\; \prod_{j=1}^7 X_j^{H_{ij}},
\qquad
S_i^Z \;:=\; \prod_{j=1}^7 Z_j^{H_{ij}},
\qquad i\in\{1,2,3\}.
\label{eq:steane_generators}
\end{equation}
A standard choice of logical Paulis is
\begin{equation}
\overline{X} := X^{\otimes 7}, \qquad \overline{Z} := Z^{\otimes 7},
\label{eq:steane_logicals}
\end{equation}
which commute with all stabilizers and anticommute with each other.

\paragraph{Logical basis.}
Let $C=\ker H$ be the $[7,4,3]$ Hamming code and let
$C^\perp=\operatorname{rowspan}(H)$.
Define the logical basis states in the usual CSS way:
\begin{equation}
\ket{0}_L := \frac{1}{\sqrt{8}}\sum_{c\in C^\perp}\ket{c},
\qquad
\ket{1}_L := \overline{X}\ket{0}_L.
\label{eq:steane_logical_states}
\end{equation}
Then $\overline{Z}\ket{\ell}_L = (-1)^{\ell}\ket{\ell}_L$ for $\ell\in\{0,1\}$.

\paragraph{The wire-permutation symmetry group.}
Index the seven data qubits by the seven nonzero vectors in $\mathbb{F}_2^3$ (i.e., by the columns
of $H$).  For any $A\in \mathrm{GL}(3,2)$, define a permutation $\pi_A$ of qubits by the rule
\begin{equation}
\text{``the column label } h_j \mapsto A h_j \text{''} \quad\Longleftrightarrow\quad
\pi_A(j)=j' \text{ such that } h_{j'} = A h_j,
\label{eq:steane_piA_def}
\end{equation}
where $h_j$ denotes the $j$th column of $H$.  This is well-defined because the columns are distinct
and $A$ maps nonzero vectors to nonzero vectors.  The resulting set
$G:=\{\pi_A : A\in \mathrm{GL}(3,2)\}$ is a subgroup of $S_7$ of size $|\mathrm{GL}(3,2)|=168$.

Permuting columns of $H$ by $\pi_A$ produces $AH$, which differs from $H$
only by an invertible row transformation. Hence conjugation by $U_{\pi_A}$
preserves the stabilizer group, though not necessarily the chosen generating
set. Moreover,
$\overline{X}=X^{\otimes 7}$ and $\overline{Z}=Z^{\otimes 7}$ are invariant under any wire
permutation, so $U_{\pi_A}$ preserves the logical basis~\eqref{eq:steane_logical_states} up to phase.
Consequently, each $U_{\pi_A}$ is monomial in the code-adapted basis $\mathcal{B}_{\mathcal{C}}$ of
Section~\ref{sec:code_specific_basis}.

\subsection{Induced action on Steane syndromes}

Write a Steane syndrome as a pair
\begin{equation}
s=(s_Z,s_X)\in\mathbb{F}_2^3\times \mathbb{F}_2^3,
\label{eq:steane_syndrome_split}
\end{equation}
where $s_Z$ records eigenvalues of the $Z$-type checks $\{S_i^Z\}$ (detecting $X$ components) and
$s_X$ records eigenvalues of the $X$-type checks $\{S_i^X\}$ (detecting $Z$ components).
For a Pauli error $E=X^a Z^b$ with $a,b\in\mathbb{F}_2^7$ (ignoring phases), the CSS commutation rules give
\begin{equation}
s_Z(E) = H a^T,\qquad s_X(E) = H b^T.
\label{eq:steane_syndrome_map}
\end{equation}

\begin{lemma}[Syndrome covariance under $\mathrm{GL}(3,2)$]
\label{lem:steane_syndrome_covariance}
Let $A\in\mathrm{GL}(3,2)$ and let $\pi_A$ be the induced qubit permutation~\eqref{eq:steane_piA_def}.
Then for every Pauli error $E$,
\begin{equation}
(s_Z(U_{\pi_A} E U_{\pi_A}^\dagger),\, s_X(U_{\pi_A} E U_{\pi_A}^\dagger))
=
(A\,s_Z(E),\, A\,s_X(E)).
\label{eq:steane_syndrome_action}
\end{equation}
\end{lemma}

\begin{proof}
Conjugating by $U_{\pi_A}$ permutes the $a,b\in\mathbb{F}_2^7$ support vectors by the same coordinate
permutation: $a\mapsto a\circ \pi_A^{-1}$ and $b\mapsto b\circ \pi_A^{-1}$.  Since $\pi_A$ permutes
the columns of $H$ according to $h_j\mapsto Ah_j$, we have the matrix identity
$H\cdot (\pi_A^{-1}\text{ acting on columns}) = A H$ over $\mathbb{F}_2$.  Therefore,
\[
H(a\circ\pi_A^{-1})^T = A(Ha^T),\qquad H(b\circ\pi_A^{-1})^T = A(Hb^T),
\]
which is exactly~\eqref{eq:steane_syndrome_action}.
\end{proof}

For $u\neq0$, let $j(u)$ be the qubit whose column of $H$ is $u$, and write
$X(u):=X_{j(u)}$, $Z(u):=Z_{j(u)}$, with $X(0)=Z(0)=I$.
Choose the syndrome representatives
\[
E_{(s_Z,s_X)}:=X(s_Z)Z(s_X),
\qquad
\ket{(s_Z,s_X),\ell}:=E_{(s_Z,s_X)}\ket{\ell}_L.
\]
These representatives are equivariant:
$U_{\pi_A}E_{(s_Z,s_X)}U_{\pi_A}^\dagger=E_{(As_Z,As_X)}$.
Since the logical codewords are invariant, the induced basis action is
$\ket{(s_Z,s_X),\ell}\mapsto\ket{(As_Z,As_X),\ell}$.

\subsection{Orbit partition and exact quotient dimension}

For the reduced QED model, the boundary state space is
\[
\overline{\mathcal{Q}} \;=\; \mathcal{B}_{\mathcal{C}} \cup \{R\},
\qquad
|\overline{\mathcal{Q}}| = 2^7+1 = 129,
\]
where $\mathcal{B}_{\mathcal{C}}=\{\ket{s,\ell}: s\in\mathbb{F}_2^3\times\mathbb{F}_2^3,\ \ell\in\{0,1\}\}$
is the code-adapted basis and $R$ is the absorbing rejection state (Section~\ref{sec:qed_lumping}).

For this example, assume that the reduced QED and gate kernels are
invariant under the induced $G\cong\mathrm{GL}(3,2)$ data-wire action.
A simple symmetric realization is a complete ideal check bracketed by
homogeneous i.i.d.\ data-Pauli noise layers,
$T_\varepsilon=N_{\mathrm{post}}T_0N_{\mathrm{pre}}$,
together with a homogeneous data-Pauli gate-noise layer, with all noise
kernels fixing $R$. Each factor respects the data-wire symmetry.
By orbit lumping (Theorem~\ref{thm:orbit_lumping_app}), the $G$-orbit
partition of $\overline{\mathcal{Q}}$ is strongly lumpable and yields an
\emph{exact} quotient chain whose dimension equals the number of orbits.

It remains to identify the orbits.  Lemma~\ref{lem:steane_syndrome_covariance} shows that $G$
acts on syndromes by the diagonal linear action
\[
(s_Z,s_X)\mapsto (A s_Z, A s_X).
\]
Since $\mathrm{GL}(3,2)$ is transitive on nonzero vectors in $\mathbb{F}_2^3$ and transitive on
ordered pairs of \emph{distinct} nonzero vectors, the orbit of a syndrome pair $(s_Z,s_X)$ is
determined entirely by the pattern
\[
(s_Z,s_X)\in\{(0,0),\ (u,0),\ (0,u),\ (u,u),\ (u,v)\ \text{with }u\neq v\},
\]
where $u,v\in\mathbb{F}_2^3\setminus\{0\}$.
With the equivariant representative choice above, the logical label $\ell$
is fixed by the action, so each syndrome orbit splits into two orbits
corresponding to $\ell\in\{0,1\}$.

\begin{proposition}[$\mathrm{GL}(3,2)$-orbit decomposition for the Steane reduced state space]
\label{prop:steane_orbit_decomp}
Let $G\cong \mathrm{GL}(3,2)$ act on $\overline{\mathcal{Q}}=\mathcal{B}_{\mathcal{C}}\cup\{R\}$ via the
data-wire permutations $\{\pi_A\}$ described above, fixing $R$.  Then the $G$-orbits in the accepted
space $\mathcal{B}_{\mathcal{C}}$ are exactly the following ten sets (five syndrome-types, each with
$\ell\in\{0,1\}$):
\begin{align*}
\mathcal{O}_{00}^{(\ell)} &:= \{\ket{(0,0),\ell}\}, \\
\mathcal{O}_{X}^{(\ell)} &:= \{\ket{(u,0),\ell} : u\in\mathbb{F}_2^3\setminus\{0\}\}, \\
\mathcal{O}_{Z}^{(\ell)} &:= \{\ket{(0,u),\ell} : u\in\mathbb{F}_2^3\setminus\{0\}\}, \\
\mathcal{O}_{Y}^{(\ell)} &:= \{\ket{(u,u),\ell} : u\in\mathbb{F}_2^3\setminus\{0\}\}, \\
\mathcal{O}_{XZ}^{(\ell)} &:= \{\ket{(u,v),\ell} : u,v\in\mathbb{F}_2^3\setminus\{0\},\ u\neq v\}.
\end{align*}
Together with the singleton orbit $\mathcal{O}_R:=\{R\}$, the orbit partition of
$\overline{\mathcal{Q}}$ has $11$ blocks.  In particular, the exact orbit-quotient chain has
dimension $11$ (down from $|\overline{\mathcal{Q}}|=129$).
\end{proposition}

\begin{proof}
The set $\{0\}$ is fixed by the linear action, and $\mathrm{GL}(3,2)$ is transitive on
$\mathbb{F}_2^3\setminus\{0\}$, so there is one orbit each for $(u,0)$ and $(0,u)$ with $u\neq 0$.
For $(u,u)$ with $u\neq 0$, transitivity again gives a single orbit.  For ordered pairs
$(u,v)$ with $u,v\neq 0$ and $u\neq v$, any such pair is linearly independent over $\mathbb{F}_2$
(because the only nonzero scalar is $1$), so there exists an invertible linear map sending
$(u,v)$ to any other ordered independent pair; hence there is a single orbit of size $7\cdot 6=42$.
Finally, the equivariant representative choice makes $\ell$ invariant, so
each syndrome-orbit splits into two disjoint orbits indexed by $\ell\in\{0,1\}$.  The rejection state
$R$ is fixed by definition of the reduced space and is a singleton orbit.  Counting orbits gives
$2\cdot 5 + 1 = 11$.
\end{proof}

\paragraph{Orbit sizes and dimension reduction.}
There are $7$ nonzero vectors in $\mathbb{F}_2^3$, so
\[
|\mathcal{O}_{00}^{(\ell)}|=1,\quad
|\mathcal{O}_{X}^{(\ell)}|=|\mathcal{O}_{Z}^{(\ell)}|=|\mathcal{O}_{Y}^{(\ell)}|=7,\quad
|\mathcal{O}_{XZ}^{(\ell)}|=7\cdot 6=42.
\]
Summing over $\ell\in\{0,1\}$ gives $2(1+7+7+7+42)=128$ accepted states, and adjoining $R$ gives
$128+1=129$ total.  Thus the symmetry/orbit quotient reduces the exact Markov model from dimension
$2^7+1=129$ down to $11$ while preserving all orbit-constant observables (Lemma~\ref{lem:observable_preserved})
and remaining exact under the symmetry assumptions (Theorem~\ref{thm:orbit_lumping_app}).

\begin{center}
\begin{tabular}{c c c}
\hline
Orbit & Interpretation (accepted-space macrostate) & Size \\
\hline
$\mathcal{O}_{00}^{(0)}$ & target logical basis state ($\ket{0}_L$) & $1$ \\
$\mathcal{O}_{00}^{(1)}$ & other logical basis state ($\ket{1}_L$) & $1$ \\
$\mathcal{O}_{X}^{(\ell)}$ & $X$-only syndrome ($s_Z\neq 0,\ s_X=0$) & $7$ \\
$\mathcal{O}_{Z}^{(\ell)}$ & $Z$-only syndrome ($s_Z=0,\ s_X\neq 0$) & $7$ \\
$\mathcal{O}_{Y}^{(\ell)}$ & matched syndrome ($s_Z=s_X\neq 0$; $Y$-type) & $7$ \\
$\mathcal{O}_{XZ}^{(\ell)}$ & mixed syndrome ($s_Z,s_X\neq 0$, $s_Z\neq s_X$) & $42$ \\
$\mathcal{O}_{R}$ & rejection $R$ (absorbing) & $1$ \\
\hline
\end{tabular}
\end{center}

\end{document}